\documentclass{article}

\usepackage{arxiv}
\usepackage{bm}
\usepackage{hyperref}
\usepackage{amsthm}
\usepackage{amsmath}
\usepackage{algorithm}
\usepackage{natbib}
\usepackage{adjustbox}
\usepackage{algorithmic}
\newtheorem{prop}{Proposition}
\usepackage[utf8]{inputenc} 
\usepackage[T1]{fontenc}    
\usepackage{hyperref}       
\usepackage{url}            
\usepackage{booktabs}       
\usepackage{amsfonts}       
\usepackage{nicefrac}       
\usepackage{microtype}      
\usepackage{lipsum}
\usepackage{graphicx}
\graphicspath{ {./images/} }

\title{Principal component-guided sparse \\reduced-rank regression}

\author{
 \href{https://orcid.org/0009-0006-4724-5058}{\includegraphics[scale=0.06]{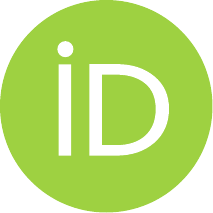}\hspace{1mm}Kanji Goto}\\
  Graduate School of Culture and Information Science, \\Doshisha University, Japan\\
\texttt{g.kanji510.510@gmail.com} \\
   \And
   \href{https://orcid.org/0000-0002-9377-1212}{\includegraphics[scale=0.06]{orcid.pdf}\hspace{1mm}Shintaro Yuki}\\
  Department of Computational and Systems Biology, \\Division of Biological Data Science, \\Medical Research Laboratory, \\Institute for Integrated Research\\
Institute of Science Tokyo, Japan\\
  \And
 \href{https://orcid.org/0000-0001-9621-5871}{\includegraphics[scale=0.06]{orcid.pdf}\hspace{1mm}Kensuke Tanioka} \\
	Department of Biomedical Sciences and Informatics\\
	Doshisha University\\
	Kyoto, Japan\\
 \And
 \href{https://orcid.org/0000-0002-1408-2655}{\includegraphics[scale=0.06]{orcid.pdf}\hspace{1mm}Hiroshi Yadohisa}\\
 Department of Culture and Information Science\\
  Doshisha University, Japan\\
}

\begin{document}
\maketitle
\begin{abstract}
Reduced-rank regression estimates regression coefficients by imposing a low-rank constraint on the matrix of regression coefficients, thereby accounting for correlations among response variables. To further improve predictive accuracy and model interpretability, several regularized reduced-rank regression methods have been proposed. However, these existing methods cannot bias the regression coefficients toward the leading principal component directions while accounting for the correlation structure among explanatory variables. In addition, when the explanatory variables exhibit a group structure, the correlation structure within each group cannot be adequately incorporated. To overcome these limitations, we propose a new method that introduces pcLasso into the reduced-rank regression framework. 
The proposed method improves predictive accuracy by accounting for the correlation among response variables while strongly biasing the matrix of regression coefficients toward principal component directions with large variance. Furthermore, even in settings where the explanatory variables possess a group structure, the proposed method is capable of explicitly incorporating this structure into the estimation process. Finally, we illustrate the effectiveness of the proposed method through numerical simulations and real data application.
\end{abstract}


\section{Introduction}
\quad In recent years, research across various fields has focused not only on analyzing the relationship between explanatory variables and a single response variable, but also on simultaneously analyzing multiple response variables to clarify their relationships. 
For example, in gene expression studies, researchers use measurements of patients' genetic variations to simultaneously predict multiple phenotypic indicators such as gene expression levels and disease status, thereby enhancing understanding of disease mechanisms and biological processes \citep{liu2022robust}.
A regression model that simultaneously predicts multiple response variables using such explanatory variables is referred to as a multivariate linear regression model \citep{izenman2008modern}. 
However, estimation of regression coefficients in multivariate linear regression model is equivalent to fitting separate univariate linear regression models for each response variable, and thus fails to capture the correlation structure among response variables \citep{hilafu2020sparse}. 
In many real datasets, including gene expression data, the response variables are often correlated \citep{chen2012sparse}. 
If such correlations are not properly accounted for, the model tends to overfit each response variable individually, leading to increased variance in the estimated regression coefficients and consequently to a decline in predictive accuracy.
\par{}
\quad To address this issue, reduced-rank regression (RRR) \citep{anderson1951estimating,izenman1975reduced} has been proposed. 
RRR imposes a low-rank constraint on the matrix of regression coefficients, allowing the variation among response variables to be explained by common latent factors, thereby capturing correlations among response variables through latent variables \citep{izenman2008modern}. 
However, when the number of explanatory variables is large or when there is strong multicollinearity among explanatory variables, the estimation accuracy of regression coefficients can deteriorate \citep{chen2012sparse,mukherjee2011reduced}. 
To overcome this problem and improve both predictive accuracy and model interpretability, regularized  RRR approaches such as sparse reduced-rank regression (SRRR) \citep{chen2012sparse} and reduced-rank regression with elastic net penalty (ERRR) \citep{kobak2021sparse} have been proposed. 
ERRR is a method that introduces a regularization term based on the elastic net \citep{zou2005regularization} to enable sparsity in the matrix of regression coefficients while accounting for the correlation structure among response variables and addressing multicollinearity among explanatory variables. 
ERRR incorporates an elastic net penalty to handle multicollinearity and enforce sparsity in the matrix of regression coefficients while accounting for correlations among response variables. 
By treating each row of the matrix of regression coefficients as a group and applying a Group Lasso \citep{yuan2006model} penalty along the row direction, we perform variable selection for explanatory variables that contribute to the predictive accuracy of the response variables.
Additionally, a penalty based on the $\ell_2$-norm of individual regression coefficients shrinks coefficients along major principal component directions corresponding to large variance in the explanatory variables, improving the stability of coefficient estimation and thus enhancing predictive accuracy.
\par{}
\quad However, this approach has two limitations. 
The first limitation is that the regression coefficients cannot be biased in a way that accounts for the correlation structure among the explanatory variables. 
As a result, the principal components that contribute to predicting the response variables cannot be adequately captured, and improvements in predictive accuracy cannot be expected. 
In particular, when the explanatory variables are highly correlated, uniformly shrinking the regression coefficients fails to sufficiently emphasize the principal component directions with large variance that reflect the correlation structure among the explanatory variables. 
Consequently, the importance of the regression coefficients becomes unclear, leading to no improvement in predictive accuracy \citep{tay2021principal}. 
The second limitation arises when explanatory variables have a group structure, as the method cannot account for within-group correlations. 
For instance, in gene expression data, genes belonging to the same functional group tend to be influenced by common genetic variations and show correlation within the group \citep{chen2010graph}. 
Ignoring such group structures in estimating regression coefficients can reduce predictive accuracy \citep{ogutu2014regularized}.
\par{}
\quad Within the framework of univariate linear regression model, Principal Components Lasso (pcLasso) \citep{tay2021principal} has been proposed as a regularization method to address these two limitations. 
pcLasso strongly biases regression coefficients toward high-variance principal components, shrinking coefficients corresponding to low-variance components and clarifying coefficient importance, thereby improving predictive accuracy. 
Moreover, pcLasso utilizes both group membership and within-group correlation information to estimate regression coefficients, further enhancing predictive accuracy. 
Specifically, principal component analysis is performed within each group to extract components reflecting within-group correlation, and regression coefficients are biased along these principal component directions. 
This approach leverages group structure to improve predictive accuracy and enables group-wise variable selection.
\par{}
\quad In this study, we propose a new method, referred to as principal component-guided sparse reduced-rank regression, which incorporates pcLasso into the RRR framework.
By imposing a low-rank constraint on the matrix of regression coefficients, the proposed method enables estimation that accounts for the correlation structure among response variables.
In addition, by strongly biasing the matrix of regression coefficients toward principal component directions with large variance, the proposed method is expected to improve predictive accuracy.
Furthermore, the proposed method enables estimation of the matrix of regression coefficients while accounting for group structures among explanatory variables. By estimating regression coefficients for each group along principal component directions that reflect within-group correlation information, further improvements in predictive accuracy are expected.
\par{}
\quad This paper demonstrates the usefulness of the proposed method through numerical simulations and real data application. 
The remainder of this paper is organized as follows. 
In Section 2, we describe the proposed model, objective function, update formulas, and estimation algorithm. 
In Section 3, we perform numerical simulations of the proposed method. 
In Section 4, we describe the application of the proposed method to the real data.
We conclude the paper and discusses future research directions in Section 5.

\section{Principal component-guided sparse reduced-rank regression model}
\quad In this Section, we first introduce the notation used to describe the proposed method and the model formulation of RRR, which forms the basis of the proposed approach. Next, we present the objective function of the proposed method. Finally, we describe the update equations and the estimation algorithm used to solve the objective function of the proposed method.

\subsection{Model formula}
\quad Here, we describe the notation and the RRR model formulation introduced in the proposed method.
Let the response matrix be $\bm{Y}=(\bm{y}_1,\bm{y}_2,\dots,\bm{y}_n)^\mathsf{T}\in \mathbb{R}^{n \times q}$, and the explanatory matrix be $\bm{X}=(\bm{x}_1,\bm{x}_2,\dots,\bm{x}_n)^\mathsf{T}\in \mathbb{R}^{n \times p}$. 
We assume a linear relationship between the response and explanatory  variables and consider the following multivariate linear regression model.
\begin{align}\label{多変量回帰モデル}
    \bm{Y}=\bm{XB}+\bm{E},
\end{align}
where $\bm{B}=(\bm{b}_1,\bm{b}_2,\dots,\bm{b}_p)^\mathsf{T}\in \mathbb{R}^{p\times q}$ 
is the matrix of regression coefficients, and $\bm{E}=(\bm{\varepsilon}_1,\bm{\varepsilon}_2,\dots,\bm{\varepsilon}_n)^\mathsf{T}\in \mathbb{R}^{n\times q}$ is the error matrix. 
In addition, we impose the following low-rank constraint on the matrix of regression coefficients $\bm{B}$, 
\begin{align}\label{ランク制約}
    rank(\bm{B})=r,\ \ \ r \leq {\rm{min}}(p,q).
\end{align}
Thus, in RRR, it is assumed that a latent low-dimensional structure exists by imposing a rank constraint on the matrix of regression coefficients $\bm{B}$.
Under this assumption, the matrix of regression coefficients $\bm{B}$ can be expressed as the product of two latent variable matrices of rank $r$ \citep{reinsel1998multivariate}, that is,
\begin{align*}
    \bm{B}=\bm{CD}^\mathsf{T},
\end{align*}
where $\bm{C}=(\bm{c}_1,\bm{c}_2,\dots,\bm{c}_p)^\mathsf{T}\in \mathbb{R}^{p \times r}$ denotes the latent variable matrix associated with the explanatory  matrix $\bm{X}$, and $\bm{D}=(\bm{d}_1,\bm{d}_2,\dots,\bm{d}_q)^\mathsf{T}\in \mathbb{R}^{q \times r}$ denotes the latent variable matrix associated with the response matrix $\bm{Y}$. Moreover,the latent variable matrix $\bm{D}$ is a column-orthonormal matrix satisfying the constraint $\bm{D}^\mathsf{T}\bm{D} = \bm{I}_r$, where $\bm{I}_r \in \mathbb{R}^{r \times r}$ is the identity matrix. 

In this case, by transforming the multivariate linear regression model in  Eq.$\eqref{多変量回帰モデル}$, the model for RRR is given by Eq.$\eqref{縮小ランク回帰}$,
\begin{equation}\label{縮小ランク回帰}
    \bm{Y}=\bm{XC}\bm{D}^\mathsf{T}+\bm{E}.
\end{equation}
\quad RRR estimates the matrix of regression coefficients by taking into account the correlations among the response variables through the low-rank constraint on the matrix of regression coefficients imposed by Eq.$\eqref{ランク制約}$ \citep{izenman2008modern}.
The linear combinations of the explanatory variables given by the $r$-dimensional reduced variables $\bm{XC}$ can be regarded as common latent factors. 
Since $\bm{XC}$ affects all response variables, the variability of the responses can be explained by these common latent factors. 
That is, the low-rank constraint enables the correlations among the response variables to be captured via the latent variables.

\subsection{Objective function}
\quad Here, we describe the objective function of the proposed method.
The objective function of the proposed method estimates $\bm{C}$ and $\bm{D}$ by solving a minimization problem consisting of a loss function based on the Frobenius norm and a regularization term. 
We assume that the $p$ explanatory variables in $\bm{X}$ are partitioned into $K$ mutually non-overlapping groups, and let $p_k \ (k=1,\dots,K)$ denote the number of explanatory variables in group $k$. 
Let $\bm{X}^{(k)} \in \mathbb{R}^{n \times p_k}$ be the explanatory matrix corresponding to group $k$, and express its singular value decomposition as $\bm{X}^{(k)} = \bm{L}^{(k)} \bm{\Sigma}^{(k)} {\bm{R}^{(k)}}^{\mathsf{T}}$. 
Under this setting, the proposed method can be written as in Eq.$\eqref{object_group}$.
\begin{align}\label{object_group}
     \underset{\bm{C}, \bm{D}} {\operatorname{min}}\ \frac{1}{2}{\|\bm{Y}-\bm{X}\bm{C}\bm{D}^\mathsf{T} \|}^2_F+\lambda\sum^K_{k=1}\sum^{p_k}_{i=1} \|\bm{C}^{(k)}_{i}\|_2
     +\frac{\theta}{2}\sum^K_{k=1}
{\rm{tr}}\left({\bm{C}^{(k)}}^\mathsf{T}\bm{A}^{(k)}\bm{C}^{(k)}\right) \qquad   {\rm{s.t.}} \qquad \bm{D}^\mathsf{T}\bm{D}=\bm{I}_r.
\end{align}
When no group structure is present in the explanatory variables, the objective function of the proposed method reduces to the case where all explanatory variables are treated as a single group, that is, $k=1$ in Eq.$\eqref{object_group}$.  
First, we explain the variables and parameters used in Eq.$\eqref{object_group}$. 
The parameters $\lambda$ and $\theta \ (\geq 0)$ are nonnegative hyperparameters determined by cross-validation (CV). 
The notation $\|\cdot\|_F$ denotes the Frobenius norm, $\|\cdot\|_2$ denotes the $\ell_2$ norm, and $\mathrm{tr}(\cdot)$ denotes the trace of a matrix. 
Moreover, $\bm{C}^{(k)} \in \mathbb{R}^{p_k \times r}$ denotes the submatrix of the latent variable matrix $\bm{C}$ corresponding to group $k$, and $\bm{C}^{(k)}_{i} \in \mathbb{R}^{1 \times r} \ (i=1,\dots,p_k)$ denotes the $i$th row vector of $\bm{C}^{(k)}$. Furthermore, $\bm{A}^{(k)} \in \mathbb{R}^{p_k \times p_k}$ is a positive semidefinite matrix that accounts for the correlation structure among the explanatory variables in group $k$, and is defined by the following expression Eq.$\eqref{semidefinite matirx}$.
\begin{align}\label{semidefinite matirx}
\bm{A}^{(k)}=\bm{R}^{(k)}\bm{\Sigma}_{(\sigma_{k_1}^2-\sigma_{k_j}^2)}{\bm{R}^{(k)}}^\mathsf{T}\qquad (k=1,\dots,K),
\end{align}
where $\bm{R}^{(k)} \in \mathbb{R}^{p_k \times p_k}$ denotes the right singular matrix obtained from the singular value decomposition of $\bm{X}^{(k)}$, and 
$\bm{\Sigma}_{(\sigma_{k_1}^{2}-\sigma_{k_j}^{2})}
=\mathrm{diag}\left(\sigma_{k_1}^{2}-\sigma_{k_1}^{2}, \sigma_{k_1}^{2}-\sigma_{k_2}^{2}, \dots, \sigma_{k_1}^{2}-\sigma_{k_{m_k}}^{2}\right)$ 
is a diagonal matrix whose diagonal elements are the differences of the squared singular values within each group. 
Here, let $\sigma_{k_j} \ (j=1,\dots,m_k)$ denote the singular values obtained from the singular value decomposition of $\bm{X}^{(k)}$, satisfying 
$\sigma_{k_1} \geq \sigma_{k_2} \geq \dots \geq \sigma_{k_{m_k}} (>0)$, where $m_k = \mathrm{rank}(\bm{X}^{(k)})$.
\par{}
\quad Next, we explain each term in Eq.$\eqref{object_group}$. 
The first term in Eq.$\eqref{object_group}$ corresponds to the objective function of RRR. 
This term enables the estimation of the matrix of regression coefficients while accounting for the correlations among the response variables. 
The second term in Eq.$\eqref{object_group}$ is a Group Lasso regularization term, where each row of $\bm{C}^{(k)}$ is regarded as a group and sparsity is imposed in a row-wise manner. 
This regularization term facilitates the selection of explanatory variables that are relevant to the response variables. 
The third term in Eq.$\eqref{object_group}$ is a regularization term based on pcLasso, constructed using the singular values within each group. 
This term strongly biases the regression coefficients toward principal components of the explanatory variables with large variance, assigning larger weights to regression coefficients corresponding to principal components with larger eigenvalues. 
The magnitude of the weighting depends on the squared differences of the singular values, $(\sigma_{k_1}^2 - \sigma_{k_j}^2)$. 
Moreover, by constructing a positive semidefinite matrix $\bm{A}^{(k)}$ for each group, the estimation of the matrix of regression coefficients reflects the correlation structure within each group. 
In applications such as gene expression data, where variables form biological or functional groups, incorporating this group structure enables coefficient estimation and variable selection at the group level and is expected to improve predictive accuracy. 
Finally, when $k=1$ and $\theta=0$, the resulting optimization problem coincides with that of SRRR \citep{chen2012sparse}. 
Therefore, the proposed method can also be regarded as an extension of SRRR.

\subsection{Estimation algorithm}
\quad Here, we describe the update equations for the parameters $\bm{C}$ and $\bm{D}$ in the objective function given by Eq.$\eqref{object_group}$ proposed in Section 2.2, as well as the estimation algorithm.
These parameters, $\bm{C}$ and $\bm{D}$, are estimated using an alternating least squares (ALS) algorithm \citep{carroll1970analysis,harshman1970foundations}, in which one parameter is fixed while the other is updated alternately.
Let $\hat{\bm{C}} \in \mathbb{R}^{p \times r}$ and $\hat{\bm{D}} \in \mathbb{R}^{q \times r}$ denote the estimators of $\bm{C}$ and $\bm{D}$ obtained by the ALS algorithm, respectively. 
Since $\bm{C}$ is estimated in a group-wise manner, we denote the submatrix of the latent variable matrix $\hat{\bm{C}}$ corresponding to group $k$ by $\hat{\bm{C}}^{(k)}$.
In order to estimate these parameters, we use the following propositions.

\begin{prop}
Fixing $\bm{D}$ as $\bm{D}^{*}$, the $i$-th row vector $\hat{\bm{C}}_{i}^{(k)} \in \mathbb{R}^{1 \times r}$ $(i = 1, \ldots, p_k)$ of the estimator $\hat{\bm{C}}^{(k)}$ is updated in order to minimize Eq.$\eqref{object_group}$ by the subgradient method \citep{friedman2007pathwise}.
\begin{align*}
    \hat{\bm{C}}_{i}^{(k)}  \leftarrow \frac{1}{{\bm{x}^{(k)}_{i}}^\mathsf{T}\bm{x}^{(k)}_i+\theta A_{ii}^{(k)}}\left(1-
    \frac{\lambda}{\|\bm{x}^{(k)}_{i}\bm{Q}_{i}-\theta \bm{s}_{i}\|_2}\right)_+
   \left({\bm{x}^{(k)}_{i}}^{\mathsf{T}}\bm{Q}_{i}-\theta\bm{s}_{i}\right),
\end{align*}
where, $\bm{x}^{(k)}_{i} \in \mathbb{R}^{n \times 1}$ denote the $i$-th column of the explanatory variable matrix $\bm{X}^{(k)}$ for group $k$,and $\bm{A}^{(k)}_{i} \in \mathbb{R}^{1 \times p_k}$ denote the $i$-th row of the positive semidefinite matrix $\bm{A}^{(k)}$.
The $i$-th diagonal element of $\bm{A}^{(k)}$ is denoted by $A^{(k)}_{ii}$. 
Furthermore, let $\tilde{\bm{C}}^{(k)}_{\ell} \, (\ell \neq i) \in \mathbb{R}^{(p_k-1) \times r}$ be the matrix obtained by excluding the $i$-th row vector $\hat{\bm{C}}^{(k)}_{i}$ from the latent variable matrix of group $k$ to be updated. Then, $\bm{Q}_{i} \in \mathbb{R}^{n \times r}$ is defined as
$\bm{Q}_{i} = \bm{Z}^{(k)} - \sum_{\ell \neq i}^{p_k} \bm{X}^{(k)}_{\ell} \tilde{\bm{C}}^{(k)}_{\ell}$, and $\bm{s}_{i} \in \mathbb{R}^{1 \times r}$ is defined as $\bm{s}_{i} = \bm{A}^{(k)}_{i} \bm{C}^{(k)} - A^{(k)}_{ii} \bm{C}^{(k)}_{i}$. Here, $\bm{Z}^{(k)}= \bm{Y}\bm{D}^{*} - \bm{X}^{(-k)} \bm{C}^{(-k)} \in \mathbb{R}^{n \times p_k}$ is the partial residual for group $k$ to be updated, where $\bm{X}^{(-k)}$ and $\bm{C}^{(-k)}$ denote the explanatory variable matrix and the latent variable matrix excluding $\bm{X}^{(k)}$ and $\bm{C}^{(k)}$, respectively, and $(z)_+ = \max(0, z)$.
\end{prop}
\begin{proof}\textit{
First, we rewrite the objective function in Eq.$\eqref{object_group}$ as a function of $\bm{C}$.
\begin{align*}
     {\|\bm{Y}-\bm{X}\bm{C}{\bm{D}^{*}}^\mathsf{T} \|}^2_F&={\rm{tr}}\left({\bm{Y}^{\mathsf{T}}\bm{Y}}\right)-2{\rm{tr}}\left({\bm{D}^{*}\bm{C}^{\mathsf{T}}\bm{X}^{\mathsf{T}}\bm{Y}}\right)+{\rm{tr}}{\left(\bm{D}^{*}\bm{C}^{\mathsf{T}}\bm{X}^{\mathsf{T}}\bm{XC}{\bm{D}^{*}}^{\mathsf{T}}\right)}\\&={\rm{tr}}\left(\bm{Y}^{\mathsf{T}}\bm{Y}\right)-  {\rm{tr}}{\left(\bm{D}^{*\mathsf{T}}\bm{D}^{*}\bm{Y}^{\mathsf{T}}\bm{Y}\right)}+
    {\rm{tr}}{\left(\bm{D}^{*\mathsf{T}}\bm{D}^{*}\bm{Y}^{\mathsf{T}}\bm{Y}\right)}\\&+{\rm{tr}}{\left(\bm{D}^{*\mathsf{T}}\bm{D}^{*}\bm{C}^{\mathsf{T}}\bm{X}^{\mathsf{T}}\bm{XC}\right)}-2{\rm{tr}}{\left(\bm{C}^{\mathsf{T}}\bm{X}^{\mathsf{T}}\bm{YD}^{*}\right)}\\&={\rm{const}}+{\rm{tr}}{\left(\bm{D}^{*\mathsf{T}}\bm{Y}^{\mathsf{T}}\bm{Y}\bm{D}^{*}\right)}+{\rm{tr}}{\left(\bm{C}^{\mathsf{T}}\bm{X}^{\mathsf{T}}\bm{XC}\right)}-2{\rm{tr}}{\left(\bm{C}^{\mathsf{T}}\bm{X}^{\mathsf{T}}\bm{YD}^{*}\right)}\\&=
    {\rm{const}}+
    {\|\bm{Y}\bm{D}^{*}-\bm{X}\bm{C} \|}^2_F,
\end{align*}
where, ${\rm const}$ represents terms independent of the parameter being updated.
Therefore, the minimization problem in Eq.$\eqref{object_group}$ can be rewritten as
Eq.$\eqref{C_object_group}$.
\begin{align}\label{C_object_group}
    \underset{\bm{C}} {\operatorname{min}}\ \  &\frac{1}{2}{\|\bm{YD}^{*}-\bm{X}\bm{C}\|}^2_F+\sum^K_{k=1}\sum^{p_k}_{i=1} \|\bm{C}^{(k)}_{i}\|_2+\frac{\theta}{2}\sum^K_{k=1}
     {\rm{tr}}\left({\bm{C}^{(k)}}^\mathsf{T}\bm{A}^{(k)}\bm{C}^{(k)}\right).
\end{align}
Fixing all parameters except for the submatrix $\bm{C}^{(k)}$ corresponding to
the $k$th group of the latent variable matrix $\bm{C}$ to be updated,
we rewrite the objective function in Eq.$\eqref{C_object_group}$ using the partial residual
$\bm{Z}^{(k)}$ for group $k$.
\begin{align}\label{sub-gradient}
  & \frac{1}{2}{\|\bm{YD}^{*}-\bm{X}^{(-k)}\bm{C}^{(-k)}-\bm{X}^{(k)}\bm{C}^{(k)}\|}^2_F+\sum^K_{k=1}\sum^{p_k}_{i=1} \|\bm{C}^{(k)}_{i}\|_2\nonumber
  +\frac{\theta}{2}\sum^K_{k=1}
{\rm{tr}}\left({\bm{C}^{(k)}}^\mathsf{T}\bm{A}^{(k)}\bm{C}^{(k)}\right)\nonumber\\&
=\frac{1}{2}{\|\bm{Z}^{(k)}-\bm{X}^{(k)}\bm{C}^{(k)}\|}^2_F+\sum^K_{k=1}\sum^{p_k}_{i=1} \|\bm{C}^{(k)}_{i}\|_2+\frac{\theta}{2}\sum^K_{k=1}
     {\rm{tr}}\left({\bm{C}^{(k)}}^\mathsf{T}\bm{A}^{(k)}\bm{C}^{(k)}\right).
\end{align}
This minimization problem can be solved using the subgradient method \citep{friedman2007pathwise}.
The update of $\bm{C}^{(k)}$ is carried out by applying a coordinate descent algorithm
to each row $\bm{C}^{(k)}_{i}$, and the subgradient of Eq.$\eqref{sub-gradient}$
with respect to $\bm{C}^{(k)}_{i}$ is derived.
\begin{align*}
    &{-\bm{x}^{(k)}_{i}}^{\mathsf{T}}\left(\bm{Q}_{i}-\bm{x}^{(k)}_{i}\bm{C}_{i}^{(k)}\right)+\lambda \bm{g}_{i} +\theta\bm{A}^{(k)}_{i}\bm{C}^{(k)}=\bm{0}_r,
\end{align*}
where, $\bm{g}_{i} \in \mathbb{R}^{1 \times r}$ denotes an element of the subdifferential
of the $\ell_{2}$ norm evaluated at $\bm{C}^{(k)}_{i}$, and is defined differently
depending on whether $\|\bm{C}^{(k)}_{i}\|_{2} = 0$ or $\|\bm{C}^{(k)}_{i}\|_{2} \neq 0$.
First, we consider the case $\|\bm{C}^{(k)}_{i}\|_{2} = 0$.
\begin{align*}
\|\bm{g}_{i}\|_2\leq1 \ \  \left(\|\bm{C}_{i}^{(k)}\|_2=0\right).
\end{align*}
Next, we consider the case $\|\bm{C}^{(k)}_{i}\|_{2} \neq 0$.
\begin{align*}
\bm{g}_{i} = \frac{\bm{C}_{i}^{(k)}}{\|\bm{C}_{i}^{(k)}\|_2}  \ \ \left(\|\bm{C}_{i}^{(k)}\|_2 \neq 0\right).
\end{align*}
Here, we consider the minimization problem by distinguishing between the cases
$\|\bm{C}^{(k)}_{i}\|_{2} = 0$ and $\|\bm{C}^{(k)}_{i}\|_{2} \neq 0$.
We first focus on the minimization problem in the case
$\|\bm{C}^{(k)}_{i}\|_{2} = 0$.
Under the condition $\|\bm{C}^{(k)}_{i}\|_{2} = 0$,
we take the subdifferential of Eq.$\eqref{C_object_group}$
with respect to $\bm{C}^{(k)}_{i}$.
\begin{align}\label{C_iの推定(0)}
    &{-\bm{x}^{(k)}_{i}}^{\mathsf{T}}\left(\bm{Q}_{i}-\bm{X}_{i}\bm{C}_{i}^{(k)}\right)+\lambda \bm{g}_{i} +\theta\bm{A}^{(k)}_{i}\bm{C}=\bm{0}_r, \nonumber\\
     \Longleftrightarrow \  &{-\bm{x}^{(k)}_{i}}^{\mathsf{T}}\left(\bm{Q}_{i}-\bm{x}^{(k)}_{i}\bm{C}_{i}^{(k)}\right)+\lambda \bm{g}_i+\theta\left(\bm{A}^{(k)}_{i}\bm{C}-A^{(k)}_{ii}\bm{C}_{i}^{(k)}+A^{(k)}_{ii}\bm{C}_{i}^{(k)}\right)=\bm{0}_r, \nonumber\\
      \Longleftrightarrow \  &{-\bm{x}^{(k)}_{i}}^{\mathsf{T}}\bm{Q}_{i}+\lambda\bm{g}_{i}+\theta\bm{s}_{i}=\bm{0}_r,
      \nonumber\\
    \Longleftrightarrow \  &\bm{g}_{i}=\frac{1}{\lambda}\left({\bm{x}^{(k)}_{i}}^{\mathsf{T}}\bm{Q}_{i}-\theta\bm{s}_{i}\right).
\end{align}
Next, we consider the minimization problem in the case $\|\bm{C}^{(k)}_{i}\|_{2} \neq 0$.
Under the condition $\|\bm{C}^{(k)}_{i}\|_{2} \neq 0$, we take the subdifferential of Eq.$\eqref{sub-gradient}$ with respect to $\bm{C}^{(k)}_{i}$.
\begin{align}\label{C_i_groupの推定}
    &-{\bm{x}^{(k)}_{i}}^{\mathsf{T}}\left(\bm{Q}_{i}-\bm{x}^{(k)}_{i}\bm{C}_{i}^{(k)}\right)+
    \lambda \frac{\bm{C}_{i}^{(k)}}{\|\bm{C}^{(k)}_{i}\|_2}+
    \theta\bm{A}^{(k)}_{i}\bm{C}^{(k)}=\bm{0}_r,\nonumber\\
    \Longleftrightarrow \  
    &-{\bm{x}^{(k)}_{i}}^{\mathsf{T}}\left(\bm{Q}_{i}-\bm{X}_{i}^{(k)}\bm{C}_{i}^{(k)}\right)+
    \lambda \frac{\bm{C}_{i}^{(k)}}{\|\bm{C}_{i}^{(k)}\|_2}+
    \theta\left(\bm{A}^{(k)}_{i}\bm{C}^{(k)}-A_{ii}\bm{C}_{i}^{(k)}+A_{ii}\bm{C}_{i}^{(k)}\right)=\bm{0}_r,\nonumber\\
    \Longleftrightarrow \  
    &-{\bm{x}^{(k)}_{i}}^{\mathsf{T}}\left(\bm{Q}_{i}-\bm{x}^{(k)}_{i}\bm{C}_{i}^{(k)}\right)+
    \lambda \frac{\bm{C}_{i}^{(k)}}{\|\bm{C}_{i}^{(k)}\|_2}+
    \theta\left(\bm{s}_{i}+A^{(k)}_{ii}\bm{C}_{i}^{(k)}\right)=\bm{0}_r, \nonumber\\
     \Longleftrightarrow \  
     &{\bm{x}^{(k)}_{i}}^{\mathsf{T}}\bm{x}^{(k)}_{i}\bm{C}_{i}^{(k)}+
    \lambda \frac{\bm{C}_{i}^{(k)}}{\|\bm{C}_{i}^{(k)}\|_2}+\theta A^{(k)}_{ii}\bm{C}_{i}^{(k)}=\left({\bm{x}^{(k)}_{i}}^{\mathsf{T}}\bm{Q}_{i}-\theta \bm{s}_{i}\right),\nonumber\\
    \Longleftrightarrow \  
    &\bm{C}_{i}^{(k)}=
    \left ({\bm{x}^{(k)}_{i}}^{\mathsf{T}}\bm{x}^{(k)}_{i}+\frac{\lambda}{\|\bm{C}_{i}^{(k)}\|_2}+\theta A^{(k)}_{ii}\right )^{-1}
    \left({\bm{x}^{(k)}_{i}}^{\mathsf{T}}\bm{Q}_{i}-\theta \bm{s}_{i}\right).
\end{align}
Next, we take the $\ell_{2}$ norm on both sides of Eq.$\eqref{C_i_groupの推定}$.
\begin{align}\label{C_group_iのL2ノルム}
     &\|\bm{C}_{i}^{(k)}\|_2=\left\|
    \left ({\bm{x}^{(k)}_{i}}^{\mathsf{T}}\bm{x}^{(k)}_{i}+\frac{\lambda}{\|\bm{C}_{i}^{(k)}\|_2}+\theta A^{(k)}_{ii}\right )^{-1}
    \left({\bm{x}^{(k)}_{i}}^{\mathsf{T}}\bm{Q}_{i}-\theta \bm{s}_{i}\right)\right\|_2, \nonumber\\
     \Longleftrightarrow \ &  {\bm{x}^{(k)}_{i}}^{\mathsf{T}}\bm{x}^{(k)}_{i}\|\bm{C}_{i}^{(k)}\|_2+\lambda+\theta A_{ii}{\|\bm{C}_{i}^{(k)}\|_2}=
    \left\|{\bm{x}^{(k)}_{i}}^{\mathsf{T}}\bm{Q}_{i}-\theta \bm{s}_{i}\right\|_2, \nonumber \\ 
     \Longleftrightarrow \ &
     \|\bm{C}_{i}^{(k)}\|_2=
    \frac{\|{\bm{x}^{(k)}_{i}}^{\mathsf{T}}\bm{Q}_{i}-\theta \bm{s}_{i}\|_2-\lambda}{{\bm{x}^{(k)}_{i}}^{\mathsf{T}}\bm{x}^{(k)}_{i}+\theta A^{(k)}_{ii}}.
\end{align}
By combining Eq.$\eqref{C_i_groupの推定}$ and Eq.$\eqref{C_group_iのL2ノルム}$,
we obtain the following expression for $\bm{C}^{(k)}_{i}$.
\begin{align}\label{C_iの推定(0以外)}
 \bm{C}_{i}^{(k)}=&
    \left ({\bm{x}^{(k)}_{i}}^{\mathsf{T}}\bm{x}^{(k)}_i+\frac{\lambda}{\|\bm{C}_{i}^{(k)}\|_2}+\theta A^{(k)}_{ii}\right )^{-1}
    \left({\bm{x}^{(k)}_{i}}^{\mathsf{T}}\bm{Q}_{i}-\theta \bm{s}_{i}\right), \nonumber\\
    \Longleftrightarrow \ 
    \bm{C}_{i}^{(k)}=&
\left({\bm{x}^{(k)}_{i}}^{\mathsf{T}}\bm{x}^{(k)}_{i}+\lambda\frac{{\bm{x}^{(k)}_{i}}^{\mathsf{T}}\bm{x}^{(k)}_{i}+\theta A^{(k)}_{ii}}{\|\bm{x}^{(k)}_{i}\bm{Q}_{i}-\theta \bm{s}_{i}\|_2-\lambda}+\theta A^{(k)}_{i i}\right)^{-1}\nonumber
\left({\bm{x}^{(k)}_{i}}^{\mathsf{T}}\bm{Q}_{i}-\theta \bm{s}_{i}\right), \nonumber\\
 \Longleftrightarrow \bm{C}_{i}^{(k)}=&
     \left\{\frac{\|{\bm{x}^{(k)}_{i}}^{\mathsf{T}}\bm{Q}_{i}-\theta \bm{s}_i\|_2\left({\bm{x}^{(k)}_{i}}^{\mathsf{T}}\bm{x}^{(k)}_{i}+\theta A_{ii}\right)}{\|{\bm{x}^{(k)}_{i}}^{\mathsf{T}}\bm{Q}_{i}-\theta \bm{s}_{i}\|_2-\lambda}\right\}^{-1}\nonumber\left({\bm{x}^{(k)}_{i}}^{\mathsf{T}}\bm{Q}_{i}-\theta \bm{s}_{i}\right),\nonumber\\
\Longleftrightarrow \ 
     \bm{C}_{i}^{(k)}=&
     \frac{1}{{\bm{x}^{(k)}_{i}}^{\mathsf{T}}\bm{x}^{(k)}_{i}+\theta A^{(k)}_{ii}}\left(1-
    \frac{\lambda}{\|{\bm{x}^{(k)}_{i}}^{\mathsf{T}}\bm{Q}_{i}-\theta \bm{s}_{i}\|_2}\right)\left({\bm{x}^{(k)}_{i}}^{\mathsf{T}}\bm{Q}_{i}-\theta \bm{s}_{i}\right).
\end{align}
Therefore, by combining Eq.$\eqref{C_iの推定(0)}$ for the case $\|\bm{C}^{(k)}_{i}\|_{2} = 0$ and Eq.$\eqref{C_iの推定(0以外)}$ for the case $\|\bm{C}^{(k)}_{i}\|_{2} \neq 0$,
we obtain Eq.$\eqref{命題1}$, which completes the proof of Proposition 1}.
\begin{align}\label{命題1}
    \bm{C}_{i}^{(k)}=
     \frac{1}{{\bm{x}^{(k)}_{i}}^{\mathsf{T}}\bm{x}^{(k)}_{i}+\theta A^{(k)}_{ii}}\left(1-
    \frac{\lambda}{\|{\bm{x}^{(k)}_{i}}^{\mathsf{T}}\bm{Q}_{i}-\theta \bm{s}_{i}\|_2}\right)_+\left({\bm{x}^{(k)}_{i}}^{\mathsf{T}}\bm{Q}_{i}-\theta \bm{s}_{i}\right).
\end{align}
\end{proof}
\begin{prop}
Fixing $\bm{C}$ as $\bm{C}^{*}$, $\hat{\bm{D}}$ is updated in order to minimize Eq.$\eqref{object_group}$: 
\begin{align*}
        {\hat{\bm{D}}} \leftarrow \bm{U}{\bm{V}}^{\mathsf{T}},
\end{align*}
where, $\bm{U} \in \mathbb{R}^{q \times h}$ and $\bm{V} \in \mathbb{R}^{r \times h}$ are the matrices composed of the left and right singular vectors obtained from the singular value decomposition of $\bm{Y}^{\mathsf{T}}\bm{X}\bm{C}^{*}$. 
Let $h$ denote the rank of $\bm{Y}^{\mathsf{T}}\bm{X}\bm{C}^*$. 
Both $\bm{U}$ and $\bm{V}$ are column-orthonormal matrices, and $\bm{\Lambda} \in \mathbb{R}^{h \times h}$ is a diagonal matrix.
\end{prop}
\begin{proof}\textit{
First, we rewrite the objective function in Eq.$\eqref{object_group}$ as a function of $\bm{D}$.
\begin{align}\label{D_object}
     \underset{\bm{D}} {\operatorname{min}}\ \frac{1}{2}{\|\bm{Y}-\bm{X}\bm{C}^{*}\bm{D}^{\mathsf{T}} \|}^2_F \qquad  {\rm{s.t.}} \qquad \bm{D}^{\mathsf{T}}\bm{D}=\bm{I}_r.
\end{align}
Therefore, the minimization problem in Eq.$\eqref{D_object}$ can be rewritten as
Eq.$\eqref{D_max}$.
\begin{align}\label{D_max}
\frac{1}{2}{\|\bm{Y}-\bm{X}\bm{C}^{*}\bm{D}^{\mathsf{T}} \|}^2_F&=\frac{1}{2}{\rm{tr}}\left({\bm{Y}^{\mathsf{T}}\bm{Y}}\right)+\frac{1}{2}{\rm{tr}}{\left(\bm{D}\bm{C}^{*\mathsf{T}}\bm{X}^{\mathsf{T}}\bm{X}\bm{C}^{*}\bm{D}^{\mathsf{T}}\right)}
-{\rm{tr}}\left({\bm{D}\bm{C}^{*\mathsf{T}}\bm{X}^{\mathsf{T}}\bm{Y}}\right) \nonumber \\
&={\rm{const}}-{\rm{tr}}\left({\bm{C}^{*\mathsf{T}}\bm{X}^{\mathsf{T}}\bm{YD}}\right).
\end{align}
Therefore, the minimization problem in Eq.$\eqref{D_max}$ can be solved by
maximizing ${\rm tr}\ \left(\bm{C}^{\mathsf{T}}\bm{X}^{\mathsf{T}}\bm{Y}\bm{D}\right)$,
which yields the estimator $\hat{\bm{D}}$.
\begin{align*}
     \underset{\bm{D}} {\operatorname{max}}\ {\rm{tr}}\left\{\left(\bm{Y}^{\mathsf{T}}\bm{X}\bm{C}^{*}\right)^{\mathsf{T}}\bm{D} \right\} \qquad  {\rm{s.t.}} \qquad \bm{D}^{\mathsf{T}}\bm{D}=\bm{I}_r.
\end{align*}
According to Theorem A.4.2 in \cite{adachi2016matrix},
this maximization problem can be solved via an orthogonal Procrustes problem \citep{gower2004procrustes}
using the singular value decomposition, which establishes Proposition 2}.
\end{proof}

\quad The estimation algorithm of the proposed method is based on an alternating least squares algorithm, which repeatedly alternates between estimating the latent variable matrices $\bm{C}$ and $\bm{D}$, and is summarized in Algorithm 1.

\begin{algorithm}
\caption{\enskip Estimation algorithm of the proposed method}\label{alg1}
\begin{algorithmic}[1]
\normalfont
\REQUIRE $\bm{X}, \bm{Y}, \lambda, \theta$
\ENSURE  $\bm{C}, \bm{D}$
\STATE Set initial value for $\bm{C}^{[0]}, \bm{D}^{[0]}$
\FOR{$k=1$ to $K$}
    \STATE $\bm{L}^{(k)}\bm{\Sigma}^{(k)}{\bm{R}^{(k)}}^{\mathsf{T}}$ $\leftarrow$ singular value decomposition of $\bm{X}^{(k)}$
    \STATE $\bm{A}^{(k)}\leftarrow \bm{R}^{(k)}\bm{\Sigma}_{(\sigma_{k_1}^2-\sigma_{k_j}^2)}{\bm{R}^{(k)}}^{\mathsf{T}}$
    \ENDFOR
    \FOR{$t=1$ to $T$}
    \FOR{$k=1$ to $K$}
    \STATE $\bm{Z}^{(k)}\leftarrow\bm{Y}\bm{D}^{[t-1]}-\bm{X}^{(-k)}{\bm{C}^{(-k)}}^{[t-1]}$
    \FOR{$i=1$ to $p_k$}
    \STATE$\bm{Q}_{i}\leftarrow \bm{Z}^{(k)}-\sum_{\ell\neq i}^{p_k} {\bm{X}}^{(k)}_{\ell}
    \tilde{\bm{C}}_{\ell}^{(k)^{[t-1]}}$
    \STATE$\bm{s}_{i}\leftarrow \bm{A}^{(k)}_{i}{\bm{C}^{(k)}}^{[t-1]}-A^{(k)}_{i i }{\bm{C}_{i}^{(k)}}^{[t-1]}$
    \STATE$\bm{C}^{{(k)}^{[t]}}_{i}
    \leftarrow \frac{1}{{\bm{x}^{(k)}_{i}}^{\mathsf{T}}{\bm{x}^{(k)}_{i}}+\theta A^{(k)}_{ii}} \left(1-\frac{\lambda}{\|{\bm{x}_{i}^{(k)}}^{\mathsf{T}}\bm{Q}_{i}-\theta \bm{s}_{i}\|_2}\right)_+\left({\bm{x}_{i}^{(k)}}^{\mathsf{T}}\bm{Q}_{i}-\theta \bm{s}_{i}\right)$
    \ENDFOR
    \ENDFOR
    \STATE $\bm{U}\bm{\Lambda}\bm{V}^{\mathsf{T}}$ $\leftarrow$ singular value decomposition of $\bm{Y}^{\mathsf{T}}\bm{X}\bm{C}^{[t]}$ 
    \STATE $\bm{D}^{[t]}\leftarrow\bm{U}\bm{V}^{\mathsf{T}}$
    \ENDFOR
    \STATE$\bm{C} \leftarrow \bm{C}^{[t]}$
    \STATE$\bm{D} \leftarrow \bm{D}^{[t]}$
    \RETURN $\bm{C}, \bm{D}$
\end{algorithmic}
\end{algorithm}

\section{Numerical Study}
\quad In this Section, we present the simulation design and the results of the numerical simulations. 
\subsection{Simulation design}
\quad Here, we describe the design of the numerical simulations. 
In the numerical simulations of this paper, we conduct simulations in which the explanatory variables do not have a group structure or have a group structure.
\par{}
\quad First, we describe the parameters used for data generation. The number of response variables $q$ is fixed at 5.
Let $p$ denote the number of explanatory variables, $n$ the sample size, $p_0$ the number of nonzero rows in the matrix of regression coefficients, and $\tau$ the proportion of nonzero rows in the matrix of regression coefficients.
That is, the number of nonzero rows $p_0$ can be expressed as $p_0 = p \cdot \tau$.
The specific values of $p$, $n$, and $\tau$ are described later.
\par{}
\quad Next, we describe the data generation process. For data generation, the rank, the number of groups, and the method for generating the explanatory variables follow \cite{tay2021principal}, with some modifications made to generate multiple response variables. We first describe the data generation procedure used in the numerical simulations for the case in which the explanatory variables do not have a group structure.
The explanatory variable matrix $\bm{X} \in \mathbb{R}^{n \times p}$ is constructed by separating it into the explanatory variable matrix $\bm{X}_{p_0} \in \mathbb{R}^{n \times p_0}$, which corresponds to the rows with nonzero elements in the matrix of regression coefficients, and the explanatory variable matrix $\bm{X}_{(p-p_0)} \in \mathbb{R}^{n \times (p-p_0)}$, which corresponds to the rows with zero elements.
The matrix $\bm{X}_{p_0}$ is generated as a matrix with correlated variables such that its rank is approximately three, using the following formula.
\begin{align*}
   \bm{X}_{p_0} = \sum_{i=1}^{3} i \cdot \bm{u}_i \bm{v}_i^{\mathsf{T}} + \bm{E}_{\bm{X}},
\end{align*}
where, $\bm{u}_i \in \mathbb{R}^{n \times 1}$ and $\bm{v}_i \in \mathbb{R}^{p_0 \times 1}$ are generated from the standard normal distribution, and $\bm{E}_{\bm{X}} \in \mathbb{R}^{n \times p_0}$ is generated from the multivariate normal distribution ${\rm N}(\bm{0}_{p_0}, \bm{I}_{p_0})$. 
Here, $\bm{0}_{p_0} \in \mathbb{R}^{p_0 \times 1}$ denotes a $p_0$-dimensional column vector whose elements are all zero.
On the other hand, $\bm{X}_{(p-p_0)}$ is generated from the multivariate normal distribution ${\rm N}(\bm{0}_{(p-p_0)}, \bm{I}_{(p-p_0)})$.
the matrix of regression coefficients is defined as $\bm{B} = \bm{C}\bm{D}^{\mathsf{T}} \in \mathbb{R}^{p \times q}$, and its rank $r$ is set to 3. 
The latent variable matrix $\bm{C}$ is partitioned as $\bm{C} = (\bm{C}_{p_0}, \bm{C}_{(p-p_0)}) \in \mathbb{R}^{p \times r}$.
The latent variable matrix $\bm{C}_{p_0} \in \mathbb{R}^{p_0 \times r}$, which contains nonzero elements, is constructed using the first three principal component vectors obtained by performing principal component analysis on $\bm{X}_{p_0}$.
The latent variable matrix $\bm{C}_{(p-p_0)} \in \mathbb{R}^{(p-p_0) \times r}$ is set to $\bm{C}_{(p-p_0)} = \bm{O}_{(p-p_0) \times r}$, where $\bm{O}$ denotes a matrix whose elements are all zero.
The latent variable matrix $\bm{D} \in \mathbb{R}^{q \times r}$ is obtained by performing a QR decomposition on a matrix generated from the distribution ${\rm N}(\bm{0}_r, \bm{I}_r)$, and using the resulting orthogonal matrix.
In addition, the error matrix $\bm{E} \in \mathbb{R}^{n \times q}$ is generated from the distribution ${\rm N}(\bm{0}_q, \bm{I}_q)$, and the response matrix $\bm{Y} \in \mathbb{R}^{n \times q}$ is constructed according to the following model.
\begin{align*}
  \bm{Y}=\bm{XC}\bm{D}^{\mathsf{T}}+\bm{E}.
\end{align*}
\par{}
\quad Next, we present the data generation procedure used in the numerical simulations for the case in which the explanatory variables have a group structure.
The explanatory variable matrix $\bm{X} \in \mathbb{R}^{n \times p}$ is partitioned into 10 non-overlapping groups indexed by $k \in \{1, \dots, 10\}$, where the number of explanatory variables in each group is denoted by $p_k \in \{p_1, \dots, p_{10}\}$. 
Let $\bm{X}^{(k)} \in \mathbb{R}^{n \times p_k}$ denote the explanatory variable matrix for group $k$, which is generated separately for each group.
The first five groups, $k \in \{1, \dots, 5\}$, contain rows with nonzero elements in the matrix of regression coefficients, whereas the remaining five groups, $k \in \{6, \dots, 10\}$, do not contain any rows with nonzero elements.
\quad Here, let $p_{k_0}$ denote the number of nonzero rows in the matrix of regression coefficients for each group, so that $p_0 = 5 p_{k_0}$. 
For $k = 1, \dots, 5$, the matrix $\bm{X}^{(k)}$ is constructed by partitioning it as $\bm{X}^{(k)} = ( \bm{X}^{(k)}_{p_{k_0}}, \bm{X}^{(k)}_{(p_k - p_{k_0})})$, where $\bm{X}^{(k)}_{p_{k_0}} \in \mathbb{R}^{n \times p_{k_0}}$ corresponds to the rows with nonzero elements in the matrix of regression coefficients, and $\bm{X}^{(k)}_{(p_k - p_{k_0})} \in \mathbb{R}^{n \times (p_k - p_{k_0})}$ corresponds to the rows with zero elements.
The matrix $\bm{X}^{(k)}_{p_{k_0}}$ is generated in the same manner.
\begin{align*}
   \bm{X}^{(k)}_{p_{k_0}} = \sum_{i=1}^{3} i \cdot \bm{u}_i \bm{v}_i^{\mathsf{T}} + \bm{E}_{\bm{X}} \ \ \ (k=1,\dots , 5).
\end{align*}
In addition, $\bm{X}^{(k)}_{(p_k - p_{k_0})}$ is generated from the distribution ${\rm N}(\bm{0}_{(p_k - p_{k_0})}, \bm{I}_{(p_k - p_{k_0})})$. 
On the other hand, for the remaining five groups, the explanatory variable matrices $\bm{X}^{(k)} \ (k = 6, \dots, 10)$ are generated from ${\rm N}(\bm{0}_{p_k}, \bm{I}_{p_k})$.
the matrix of regression coefficients is generated in the same manner, and the latent variable matrix $\bm{C}$ is generated separately for each group, where $\bm{C}^{(k)} \in \mathbb{R}^{p_k \times r}$ $(k = 1, \ldots, 10)$ denotes the latent variable matrix corresponding to group $k$. 
The first five groups contain nonzero elements, while the remaining five groups contain no nonzero elements.
For the first five groups, the latent variable matrices $\bm{C}^{(k)} \ (k = 1, \ldots, 5)$ are generated by partitioning $\bm{C}^{(k)} = (\bm{C}^{(k)}_{p_{k_0}}, \, \bm{C}^{(k)}_{(p_k - p_{k_0})})$, where $\bm{C}^{(k)}_{p_{k_0}} \in \mathbb{R}^{p_{k_0} \times r}$ consists of rows with nonzero elements, and
$\bm{C}^{(k)}_{(p_k - p_{k_0})} \in \mathbb{R}^{(p_k - p_{k_0}) \times r}$ consists of rows with zero elements.
For the latent variable matrix $\bm{C}^{(k)}_{p_{k_0}}$, principal component analysis is performed on the explanatory variable matrix $\bm{X}^{(k)}_{p_{k_0}}$ corresponding to the rows with nonzero elements, and the matrix is constructed using the principal component vectors up to the third principal component.
For the latent variable matrix $\bm{C}^{(k)}_{(p_k - p_{k_0})}$, we set $\bm{C}^{(k)}_{(p_k - p_{k_0})} = \bm{O}_{(p_k - p_{k_0}) \times r}$.
For the latent variable matrices of the remaining five groups, $\bm{C}^{(k)} \ (k = 6, \dots, 10)$, we set $\bm{C}^{(k)} = \bm{O}_{p_k \times r}$.
The latent variable matrix $\bm{D}$, the error matrix $\bm{E}$, and the response matrix $\bm{Y}$ are generated in the same manner as in the data generation procedure for the case without group structure.
\par{}
\quad Next, we describe the configuration of each scenario. In the numerical simulations, we consider the following scenarios with respect to the number of explanatory variables (Factor 1), the sample size (Factor 2), and the proportion of nonzero rows in the matrix of regression coefficients (Factor 3).
In this paper, we conduct numerical simulations over a total of 36 scenarios, obtained by combining $2$ (whether the explanatory variables have a group structure) $\times\ 2$ (Factor 1) $\times\ 3$ (Factor 2) $\times\ 3$ (Factor 3).\\

\par{}
\noindent \textbf{(Factor 1)}\\
\noindent{The number of explanatory variables is set as $p=200$ and $400$.}\\

\noindent \textbf{(Factor 2)}\\
\noindent{Sample size is set as $n,n_{\rm test}=100,500, $ and $1000$.}\\

\noindent \textbf{(Factor 3)}\\
\noindent{$\tau$ is proportion of nonzero rows in the matrix of regression coefficients, and set as $\tau=0.1,0.2$ and $0.25$.\\
\par{}
\quad Next, we describe the evaluation criteria used in the numerical simulations. 
In this study, the performance of the proposed and comparison methods is evaluated using the mean squared error (MSE) of the matrix of regression coefficients, the MSE of the predicted values, the true positive rate (TPR), and the true negative rate (TNR).
The MSE of the matrix of regression coefficients measures the discrepancy between the estimated regression coefficient matrix $\hat{\bm{B}}$ obtained by each method and the true regression coefficient matrix $\bm{B}$. In addition, the MSE of the predicted values evaluates how closely the predicted responses for the test data obtained by each method match the true response matrix $\bm{Y}$.
Hereafter, the MSE of the matrix of regression coefficients and the MSE of the predicted values are denoted by ${\rm MSE}_{\bm{B}}$ and ${\rm MSE}_{\bm{Y}}$, respectively, and their definitions are given as follows.
\begin{align*}
    {\rm{MSE}}_{{\bm{B}}}=\frac{1}{pq}\|(\hat{\bm{C}}\hat{\bm{D}}^{\mathsf{T}}-\bm{CD}^{\mathsf{T}})\|_F^2, 
\end{align*}
\begin{align*}
    {\rm{MSE}}_{{\bm{Y}}}=\frac{1}{n_{\rm test}q}\|\bm{X}_{test}\hat{\bm{C}}\hat{\bm{D}}^{\mathsf{T}}-\bm{X}_{test}\bm{CD}^{\mathsf{T}}\|_F^2,
\end{align*}
where $\bm{X}_{test} \in \mathbb{R}^{n_{\rm test} \times p}$ is the explanatory variable matrix for the test data. 
Next, TPR and TNR are metrics used to evaluate how accurately the explanatory variables are selected. 
TPR represents the proportion of rows in the matrix of regression coefficients that are nonzero and correctly estimated, whereas TNR represents the proportion of rows that are zero and correctly estimated.
\par{}
\quad Next, we describe the comparative methods used in the numerical simulations. The four comparative methods are Multivariate lasso (MLasso), which extends the lasso to multivariate linear regression, Multivariate elastic net (MElastic), which extends the elastic net to multivariate linear regression, SRRR, and ERRR.
Here, $\lambda$ and $\alpha$ denote hyperparameters, $\bm{B}_i \in \mathbb{R}^{1 \times q}$ denotes the $i$-th row vector of the matrix of regression coefficients $\bm{B}$. 
\begin{align*}
\intertext{(1) MLasso}
     \underset{\bm{B}} {\operatorname{min}}\ \  \frac{1}{2}{\|\bm{Y}-\bm{X}\bm{B}\|}^2_2+\lambda\sum^p_{i=1}\|\bm{B}_i\|_2 .
\end{align*}
\begin{align*}
\intertext{(2) MElastic}
     \underset{\bm{B}} {\operatorname{min}}\ \  \frac{1}{2}{\|\bm{Y}-\bm{X}\bm{B}\|}^2_2+\lambda\left(\alpha\sum^p_{i=1}\|\bm{B}_i\|_2+\frac{(1-\alpha)}{2}\|\bm{B}\|_2^2\right).
\end{align*}
\begin{align*}
\intertext{(3) SRRR}
     \underset{\bm{C},\ \bm{D}} {\operatorname{min}}\ \  \frac{1}{2}{\|\bm{Y}-\bm{X}\bm{C}\bm{D}^{\mathsf{T}} \|}^2_F+\lambda\sum^p_{i=1} \|\bm{C}_{i}\|_2 \qquad  {\rm{s.t.}} \qquad \bm{D}^{\mathsf{T}}\bm{D}=\bm{I}_r.
\end{align*}
\begin{align*}
\intertext{(4) ERRR}
     \underset{\bm{C},\ \bm{D}} {\operatorname{min}}\ \  \frac{1}{2}{\|\bm{Y}-\bm{X}\bm{C}\bm{D}^{\mathsf{T}} \|}^2_F+  \lambda\left(\alpha\sum^p_{i=1} \|\bm{C}_{i}\|_2+\frac{(1-\alpha)}{2}\|\bm{C}\|_2^2\right) \qquad
    {\rm{s.t.}} \qquad \bm{D}^{\mathsf{T}}\bm{D}=\bm{I}_r.
\end{align*}
\quad Finally, we describe the evaluation procedure for the numerical studies. For each scenario, training and test data are generated, and the hyperparameters ($\lambda, \alpha, \theta$) for each method are determined using 5-fold cross-validation on the training data. 
For the proposed method, SRRR, and ERRR, the rank $r$ is set to the true rank of 3. 
Each method is then applied to the test data using the determined hyperparameters, and the evaluation metrics for each method are calculated. 
For all 36 scenarios, data are randomly generated, and the calculation of evaluation metrics for each method is repeated 100 times to assess the performance of the proposed method.

\subsection{Simulation results}
\quad Here, we present the results of the numerical simulations. 
First, we describe the results for each scenario in the case where the explanatory variables do not have a group structure. 
Next, we describe the results for each scenario in the case where the explanatory variables have a group structure.
\subsubsection{Results of the numerical simulations in the case without a group structure}
\quad Here, we describe the results of the numerical simulations for each scenario in the case where the explanatory variables do not have a group structure. When the number of explanatory variables is $p=200$, the numbers of nonzero rows are $p_0 = \{20, 40, 50\}$, and when the number of explanatory variables is $p=400$, the numbers of nonzero rows are $p_0 = \{40, 80, 100\}$.
\par{}
\quad First, we describe the results of the numerical simulations for the case without a group structure when $p = 200$. The results are summarized in Table 1.
The proposed method achieved the lowest ${\rm MSE}_{\bm{Y}}$ in all scenarios, demonstrating higher predictive accuracy compared to the comparative methods. 
Furthermore, as the proportion of nonzero rows increases, MLasso, MElastic, and ERRR tend to exhibit higher prediction errors, whereas the proposed method maintains robust estimation accuracy. 
Regarding the selection accuracy measured by TPR, the proposed method, MElastic, and ERRR showed high selection accuracy. 
With respect to the selection accuracy measured by TNR, the proposed method and SRRR generally demonstrated high accuracy, while MElastic and ERRR showed lower TNR, indicating a tendency toward over-selection. 
These results suggest that the proposed method achieves high selection accuracy in terms of both TPR and TNR, enabling stable variable selection.
\par{}
\quad Next, we describe the results of the numerical simulations for the case without a group structure when $p = 400$. The results are summarized in Table 2.
In the scenarios with $n, n_{\rm test} = 100, 500$, the proposed method achieved the lowest ${\rm MSE}_{\bm{Y}}$, demonstrating higher predictive accuracy than the comparative  methods even in high-dimensional settings. 
In contrast, in the scenario with $n, n_{\rm test} = 1000$, SRRR exhibited the highest predictive accuracy. 
Regarding selection accuracy measured by TPR, the proposed method, MElastic, and ERRR showed high values, successfully selecting the truly nonzero rows. 
With respect to TNR, the proposed method and SRRR demonstrated relatively high values, accurately selecting the truly zero rows. 
Overall, the proposed method detects truly nonzero variables with high accuracy while suppressing the false selection of unnecessary variables.
\subsubsection{Results of the numerical simulations in the case with a group structure}
\quad Here, we describe the results of each scenario in the numerical simulations for the case where the explanatory variables have a group structure.
When the number of explanatory variables is $p=200$, the numbers of nonzero rows in each group are $p_{k_0} = \{4, 8, 10\}$, whereas when the number of explanatory variables is $p=400$, the numbers of nonzero rows in each group are $p_{k_0} = \{8, 16, 20\}$.
\par{}
\quad First, we describe the results of the numerical simulations for the case with a group structure when $p=200$. The results are summarized in Table 3.
The proposed method achieved the lowest ${\rm MSE}_{\bm{Y}}$ across all scenarios, demonstrating high predictive accuracy. 
SRRR also showed relatively high predictive accuracy, but the difference in predictive accuracy between SRRR and the proposed method tended to increase as the proportion of nonzero rows increased. 
Regarding selection accuracy, MElastic and ERRR exhibited high TPR values, but many scenarios showed a decrease in TNR, indicating a tendency to over-select nonzero rows. 
In contrast, the proposed method demonstrated high selection accuracy in terms of both TPR and TNR, and was able to stably select variables, particularly for moderate to large sample sizes. 
These results indicate that the proposed method, which accounts for the group structure of the explanatory variables, can accurately select important variables while maintaining a low MSE.
\par{}
\quad Next, we describe the results of the numerical simulations for the case with a group structure when $p=400$. The results are summarized in Table 4.
In the scenarios with $n, n_{\rm test} = 100, 500$, the proposed method achieved the lowest ${\rm MSE}_{\bm{Y}}$, demonstrating consistently high predictive accuracy even in high-dimensional scenarios with a group structure. 
In the scenario with $n, n_{\rm test} = 1000$, SRRR achieved the smallest ${\rm MSE}_{\bm{Y}}$, but the difference from the proposed method was small, yielding comparable results. 
Regarding selection accuracy, the proposed method exhibited relatively high TPR, whereas TNR was higher for the comparative methods in many scenarios. 
These results indicate that even in scenarios with larger $p$, the proposed method, which accounts for the group structure of the explanatory variables, effectively leverages the group structure to improve predictive accuracy.
\newpage
\begin{table}[H]
\centering
\caption{Results of the numerical simulations for the case without a group structure when the number of explanatory variables is $p = 200$}
\resizebox{1.0\textwidth}{!}{
\begin{tabular}{|l|cccc|cccc|cccc|}
\hline
 & \multicolumn{4}{c|}{$n,n_{\rm test}=100,\ \tau=0.1$}
 & \multicolumn{4}{c|}{$n,n_{\rm test}=100,\ \tau=0.2$}
 & \multicolumn{4}{c|}{$n,n_{\rm test}=100,\ \tau=0.25$} \\ \hline
\textbf{\textit{Method}} & $\bm{\mathit{MSE}_{\bm{B}}}$ & $\bm{\mathit{MSE}_{\bm{Y}}}$ & \textbf{\textit{TPR}} & \textbf{\textit{TNR}} & $\bm{\mathit{MSE}_{\bm{B}}}$ & $\bm{\mathit{MSE}_{\bm{Y}}}$ & \textbf{\textit{TPR}} & \textbf{\textit{TNR}} & $\bm{\mathit{MSE}_{\bm{B}}}$ & $\bm{\mathit{MSE}_{\bm{Y}}}$ & \textbf{\textit{TPR}} & \textbf{\textit{TNR}}  \\
\hline
Proposed
 & \textbf{0.0004} & \textbf{0.1049} & 0.9433 & 0.9019
 & \textbf{0.0004} & \textbf{0.1174} & 0.9325 & 0.9204
 & \textbf{0.0005} & \textbf{0.1300} & 0.9027 & 0.9511 \\
MLasso
 & 0.0005 & 0.1842 & 0.9633 & 0.8274
 & 0.0007 & 0.2257 & 0.9208 & 0.9281
 & 0.0010 & 0.2362 & 0.8973 & 0.8516 \\
MElastic
 & 0.0007 & 0.2563 & \textbf{0.9917} & 0.7633
 & 0.0008 & 0.2526 & \textbf{0.9833} & 0.7815
 & 0.0009 & 0.2540 & \textbf{0.9853} & 0.7949 \\
SRRR
 & \textbf{0.0004} & 0.1111 & 0.9400 & \textbf{0.9056}
 & 0.0008 & 0.1521 & 0.8408 & \textbf{0.9538}
 & 0.0009 & 0.1744 & 0.7873 & \textbf{0.9613} \\
ERRR
 & 0.0006 & 0.2261 & \textbf{0.9917} & 0.8744
 & 0.0007 & 0.2243 & 0.9792 & 0.8921
 & 0.0009 & 0.2120 & 0.9807 & 0.7442 \\
\hline
  & \multicolumn{4}{c|}{$n,n_{\rm test}=500,\ \tau=0.1$}
 & \multicolumn{4}{c|}{$n,n_{\rm test}=500,\ \tau=0.2$}
 & \multicolumn{4}{c|}{$n,n_{\rm test}=500,\ \tau=0.25$} \\ \hline  
     \textbf{\textit{Method}} & $\bm{\mathit{MSE}_{\bm{B}}}$ & $\bm{\mathit{MSE}_{\bm{Y}}}$ & \textbf{\textit{TPR}} & \textbf{\textit{TNR}} & $\bm{\mathit{MSE}_{\bm{B}}}$ & $\bm{\mathit{MSE}_{\bm{Y}}}$ & \textbf{\textit{TPR}} & \textbf{\textit{TNR}} & $\bm{\mathit{MSE}_{\bm{B}}}$ & $\bm{\mathit{MSE}_{\bm{Y}}}$ & \textbf{\textit{TPR}} & \textbf{\textit{TNR}} 
     \\ \hline
    Proposed & \textbf{0.0001} &\textbf{0.0233} &0.9950& 0.7672 &\textbf{0.0001}& \textbf{0.0314} &0.9950& 0.7452&\textbf{0.0001} & \textbf{0.0404}& \textbf{1.0000} &0.5060\\
     MLasso & 0.0002 &0.0507 &0.9933 &0.6357 &0.0003 &0.0660 &0.9917 &0.6458 &
      0.0003 &0.0737 &0.9880 &0.6469\\
     MElastic & 0.0003& 0.0817 &\textbf{0.9967}& 0.4044  &0.0004& 0.0934 &\textbf{0.9958}& 0.4150&0.0004 &0.0975 &0.9953 &0.4164\\
         SRRR & \textbf{0.0001} &0.0268 &0.9783 &\textbf{0.8959}& 0.0002 &0.0447& 0.9833 &\textbf{0.8583}&0.0003 &0.0527& 0.9520& \textbf{0.9060}\\
     ERRR & 0.0002 &0.0581 &\textbf{0.9967} &0.6683 & 0.0003 &0.0708& 0.9942 &0.6775&0.0003 &0.0748 &0.9933 &0.6771\\
    \hline
       & \multicolumn{4}{c|}{$n,n_{\rm test}=1000,\ \tau=0.1$}
 & \multicolumn{4}{c|}{$n,n_{\rm test}=1000,\ \tau=0.2$}
 & \multicolumn{4}{c|}{$n,n_{\rm test}=1000,\ \tau=0.25$} \\ \hline 
     \textbf{\textit{Method}} & $\bm{\mathit{MSE}_{\bm{B}}}$ & $\bm{\mathit{MSE}_{\bm{Y}}}$ & \textbf{\textit{TPR}} & \textbf{\textit{TNR}} & $\bm{\mathit{MSE}_{\bm{B}}}$ & $\bm{\mathit{MSE}_{\bm{Y}}}$ & \textbf{\textit{TPR}} & \textbf{\textit{TNR}} & $\bm{\mathit{MSE}_{\bm{B}}}$ & $\bm{\mathit{MSE}_{\bm{Y}}}$ & \textbf{\textit{TPR}} & \textbf{\textit{TNR}} 
     \\ \hline
    Proposed & \textbf{0.0001}& \textbf{0.0148} &\textbf{0.9983}& 0.7137& \textbf{0.0001}& \textbf{0.0215}& 0.9942 &\textbf{0.9565}& \textbf{0.0001}& \textbf{0.0278} &0.9960 &\textbf{0.9396}\\
     MLasso & \textbf{0.0001} &0.0280 &0.9967 &0.7470& 0.0002 &0.0370 &0.9975 &0.6008&0.0002 &0.0422 &0.9933 &0.6062\\
     MElastic & 0.0002 &0.0480 &\textbf{0.9983}& 0.7330& 0.0002& 0.0527& \textbf{0.9983}& 0.7404&
     0.0002 &0.0562 &\textbf{0.9973} &0.7420\\
     SRRR & \textbf{0.0001} &0.0151 &0.9967 &0.8350&\textbf{0.0001} &0.0244& 0.9883& 0.8333&\textbf{0.0001} &0.0294& 0.9860 &0.7780\\
     ERRR & \textbf{0.0001} &0.0435 &\textbf{0.9983} &\textbf{0.8980}& 0.0002 &0.0472& 0.9942& 0.9050& 0.0002 &0.0509& 0.9953& 0.9040\\
     \hline
\end{tabular}}
\end{table}
\begin{table}[H]
\centering
\caption{Results of the numerical simulations for the case without a group structure when the number of explanatory variables is $p = 400$}
\resizebox{1.0\textwidth}{!}{
\begin{tabular}{|l|cccc|cccc|cccc|}
\hline
 & \multicolumn{4}{c|}{$n,n_{\rm test}=100,\ \tau=0.1$}
 & \multicolumn{4}{c|}{$n,n_{\rm test}=100,\ \tau=0.2$}
 & \multicolumn{4}{c|}{$n,n_{\rm test}=100,\ \tau=0.25$} \\ \hline
\textbf{\textit{Method}} & $\bm{\mathit{MSE}_{\bm{B}}}$ & $\bm{\mathit{MSE}_{\bm{Y}}}$ & \textbf{\textit{TPR}} & \textbf{\textit{TNR}} & $\bm{\mathit{MSE}_{\bm{B}}}$ & $\bm{\mathit{MSE}_{\bm{Y}}}$ & \textbf{\textit{TPR}} & \textbf{\textit{TNR}} & $\bm{\mathit{MSE}_{\bm{B}}}$ & $\bm{\mathit{MSE}_{\bm{Y}}}$ & \textbf{\textit{TPR}} & \textbf{\textit{TNR}}  \\
\hline
 Proposed & \textbf{0.0004}& \textbf{0.1741} &0.9283 &0.8256& 0.0008& \textbf{0.2347} &0.7800 &0.9534 &0.0012 &\textbf{0.2550} &0.7477 &0.9502\\
     MLasso & \textbf{0.0004} &0.2308 &0.9175 &\textbf{0.9431}& 0.0007& 0.2984& 0.8175 &\textbf{0.9574} &0.0009 &0.3084 &0.8163 &0.9177\\
     MElastic & \textbf{0.0004}& 0.2929& \textbf{0.9817}& 0.9407& 0.0006& 0.2850& \textbf{0.9604} &0.8722 &0.0006 &0.2808& \textbf{0.9660}& 0.8810\\
     SRRR & 0.0005 &0.2069 &0.8783 &0.8574&0.0007 &0.2328 &0.7458 &0.9257 
     &0.0010 &0.3290 &0.7580 &0.8330\\
     ERRR & \textbf{0.0004} &0.2298& 0.9783 &0.9169& \textbf{0.0005} &0.2465 &0.9529 &0.9384 & \textbf{0.0005} &0.2580& 0.9570 &\textbf{0.9719}\\
\hline
  & \multicolumn{4}{c|}{$n,n_{\rm test}=500,\ \tau=0.1$}
 & \multicolumn{4}{c|}{$n,n_{\rm test}=500,\ \tau=0.2$}
 & \multicolumn{4}{c|}{$n,n_{\rm test}=500,\ \tau=0.25$} \\ \hline  
      \textbf{\textit{Method}} & $\bm{\mathit{MSE}_{\bm{B}}}$ & $\bm{\mathit{MSE}_{\bm{Y}}}$ & \textbf{\textit{TPR}} & \textbf{\textit{TNR}} & $\bm{\mathit{MSE}_{\bm{B}}}$ & $\bm{\mathit{MSE}_{\bm{Y}}}$ & \textbf{\textit{TPR}} & \textbf{\textit{TNR}} & $\bm{\mathit{MSE}_{\bm{B}}}$ & $\bm{\mathit{MSE}_{\bm{Y}}}$ & \textbf{\textit{TPR}} & \textbf{\textit{TNR}} 
     \\ \hline
   Proposed & \textbf{0.0001} &\textbf{0.0475} &0.9842 &0.7949& \textbf{0.0001}& \textbf{0.0571}& \textbf{0.9950}& \textbf{0.8966} & \textbf{0.0001}& \textbf{0.0814} &\textbf{0.9953}& \textbf{0.9057}\\
     MLasso & 0.0002 &0.0774 &0.9900 &0.6906&0.0002& 0.1016 &0.9733& 0.7032
    & 0.0003 &0.1144 &0.9597 &0.6961\\
     MElastic &  0.0002& 0.0954 &\textbf{0.9925}& \textbf{0.8229}& 0.0002 &0.1060& 0.9850 &0.8291 & 0.0003 &0.1137 &0.9867 &0.8299\\
     SRRR & \textbf{0.0001} &0.0585& 0.9825 &0.7408&0.0002 &0.0792& 0.9462& 0.8021
     &0.0002 &0.0882 &0.9257& 0.8303\\
     ERRR & 0.0002 &0.0796& 0.9950 &0.7142&0.0002& 0.0930 &0.9871& 0.7196
     &0.0003& 0.1026& 0.9877 &0.7240\\
    \hline
       & \multicolumn{4}{c|}{$n,n_{\rm test}=1000,\ \tau=0.1$}
 & \multicolumn{4}{c|}{$n,n_{\rm test}=1000,\ \tau=0.2$}
 & \multicolumn{4}{c|}{$n,n_{\rm test}=1000,\ \tau=0.25$} \\ \hline 
     \textbf{\textit{Method}} & $\bm{\mathit{MSE}_{\bm{B}}}$ & $\bm{\mathit{MSE}_{\bm{Y}}}$ & \textbf{\textit{TPR}} & \textbf{\textit{TNR}} & $\bm{\mathit{MSE}_{\bm{B}}}$ & $\bm{\mathit{MSE}_{\bm{Y}}}$ & \textbf{\textit{TPR}} & \textbf{\textit{TNR}} & $\bm{\mathit{MSE}_{\bm{B}}}$ & $\bm{\mathit{MSE}_{\bm{Y}}}$ & \textbf{\textit{TPR}} & \textbf{\textit{TNR}} 
     \\ \hline
    Proposed & \textbf{0.0001} &0.0305 &0.9917& 0.8889& \textbf{0.0001} &0.0562 &\textbf{0.9983} &0.8660&\textbf{0.0001} &0.0687& \textbf{0.9987} &0.6613\\
     MLasso & \textbf{0.0001} &0.0399 &0.9883& 0.9258&  \textbf{0.0001} &0.0590& 0.9850 &\textbf{0.9295}&0.0002 &0.0680 &0.9763 &  \textbf{0.9301}\\
     MElastic & \textbf{0.0001}& 0.0724& 0.9925& \textbf{0.9837}& 0.0002 &0.0689 &0.9958 &0.7594& 0.0002 &0.0738 &0.9923& 0.7576\\
     SRRR &  \textbf{0.0001} &\textbf{0.0254}& 0.9858 &0.8894& \textbf{0.0001}& \textbf{0.0436}& 0.9767& 0.7958&\textbf{0.0001}& \textbf{0.0507} &0.9697 &0.8126\\
     ERRR & \textbf{0.0001} &0.0480& \textbf{0.9933} &0.9080 &  \textbf{0.0001}&0.0606& 0.9942 &0.9100&0.0002 &0.0662& 0.9900 &0.9124\\
     \hline
\end{tabular}}
\end{table}
\begin{table}[H]
\centering
\caption{Results of the numerical simulations for the case with a group structure when the number of explanatory variables is $p = 200$}
\resizebox{1.0\textwidth}{!}{
\begin{tabular}{|l|cccc|cccc|cccc|}
\hline
 & \multicolumn{4}{c|}{$n,n_{\rm test}=100,\ \tau=0.1$}
 & \multicolumn{4}{c|}{$n,n_{\rm test}=100,\ \tau=0.2$}
 & \multicolumn{4}{c|}{$n,n_{\rm test}=100,\ \tau=0.25$} \\ \hline
 \textbf{\textit{Method}} & $\bm{\mathit{MSE}_{\bm{B}}}$ & $\bm{\mathit{MSE}_{\bm{Y}}}$ & \textbf{\textit{TPR}} & \textbf{\textit{TNR}} & $\bm{\mathit{MSE}_{\bm{B}}}$ & $\bm{\mathit{MSE}_{\bm{Y}}}$ & \textbf{\textit{TPR}} & \textbf{\textit{TNR}} & $\bm{\mathit{MSE}_{\bm{B}}}$ & $\bm{\mathit{MSE}_{\bm{Y}}}$ & \textbf{\textit{TPR}} & \textbf{\textit{TNR}}  \\
\hline
  Proposed & \textbf{0.0004}& \textbf{0.1001} &0.9410 & 0.9042 &\textbf{0.0004}& \textbf{0.1170} &0.9348& 0.9174 &\textbf{0.0005}& \textbf{0.1290}& 0.9078 & 0.9449\\
     MLasso & 0.0005 & 0.1785 & 0.9630 & 0.8285 & 0.0007 &0.2284 &0.9135 &0.9302
     &0.0010& 0.2293 & 0.8972 &0.8480 \\
     MElastic & 0.0006& 0.2657 &\textbf{0.9905} &0.8505& 0.00010 & 0.2500 & \textbf{0.9880} &0.6802 &0.0009& 0.2475 & \textbf{0.9844}& 0.7904\\
     SRRR & \textbf{0.0004} &0.1067 &0.9305 &\textbf{0.9072} &0.0007 &0.1486 &0.8340 &\textbf{0.9571} &0.0009& 0.1685 &0.7942 &\textbf{0.9627}\\
     ERRR & 0.0006 &0.1998 &\textbf{0.9905}& 0.8083& 0.0007 &0.2207 &0.9822 &0.8939
    & 0.0011 &0.2277 & 0.9820 &0.6226\\
\hline
  & \multicolumn{4}{c|}{$n,n_{\rm test}=500,\ \tau=0.1$}
 & \multicolumn{4}{c|}{$n,n_{\rm test}=500,\ \tau=0.2$}
 & \multicolumn{4}{c|}{$n,n_{\rm test}=500,\ \tau=0.25$} \\ \hline  
      \textbf{\textit{Method}} & $\bm{\mathit{MSE}_{\bm{B}}}$ & $\bm{\mathit{MSE}_{\bm{Y}}}$ & \textbf{\textit{TPR}} & \textbf{\textit{TNR}} & $\bm{\mathit{MSE}_{\bm{B}}}$ & $\bm{\mathit{MSE}_{\bm{Y}}}$ & \textbf{\textit{TPR}} & \textbf{\textit{TNR}} & $\bm{\mathit{MSE}_{\bm{B}}}$ & $\bm{\mathit{MSE}_{\bm{Y}}}$ & \textbf{\textit{TPR}} & \textbf{\textit{TNR}} 
     \\ \hline
    Proposed & \textbf{0.0001} &\textbf{0.0241} &0.9970& 0.7720 &\textbf{0.0001}& \textbf{0.0316} &0.9970& 0.7431 &\textbf{0.0001} & \textbf{0.0398}& 0.9900 &0.5116\\
     MLasso & 0.0002 &0.0509 &0.9965 &0.6405 &0.0003 &0.0656 &0.9900 &0.6469 &
    0.0003 &0.0788 &0.9734 &\textbf{0.9560} \\
     MElastic & 0.0002 & 0.0799 &0.9975 & 0.7933  &0.0003 & 0.0893 &0.9932& 0.8060 &0.0004 &0.0953 &\textbf{0.9956} &0.4134\\
     SRRR & \textbf{0.0001} &0.0275 &0.9870 &\textbf{0.8968}& 0.0002 &0.0446 & 0.9730 &\textbf{0.8657}&0.0003 &0.0528& 0.9532& 0.9143\\
     ERRR & 0.0002 &0.0588 &\textbf{0.9985} &0.6720 & 0.0003 &0.0685 & \textbf{0.9952} &0.6776&0.0003 &0.0862 &0.9904 &0.9302\\
    \hline
       & \multicolumn{4}{c|}{$n,n_{\rm test}=1000,\ \tau=0.1$}
 & \multicolumn{4}{c|}{$n,n_{\rm test}=1000,\ \tau=0.2$}
 & \multicolumn{4}{c|}{$n,n_{\rm test}=1000,\ \tau=0.25$} \\ \hline 
     \textbf{\textit{Method}} & $\bm{\mathit{MSE}_{\bm{B}}}$ & $\bm{\mathit{MSE}_{\bm{Y}}}$ & \textbf{\textit{TPR}} & \textbf{\textit{TNR}} & $\bm{\mathit{MSE}_{\bm{B}}}$ & $\bm{\mathit{MSE}_{\bm{Y}}}$ & \textbf{\textit{TPR}} & \textbf{\textit{TNR}} & $\bm{\mathit{MSE}_{\bm{B}}}$ & $\bm{\mathit{MSE}_{\bm{Y}}}$ & \textbf{\textit{TPR}} & \textbf{\textit{TNR}} 
     \\ \hline
   Proposed & \textbf{0.0001}& \textbf{0.0147} &0.9975& 0.7141& \textbf{0.0001}& \textbf{0.0210}& 0.9962 &\textbf{0.9563}& \textbf{0.0001}& \textbf{0.0282} &0.9964 &\textbf{0.9383}\\
     MLasso & \textbf{0.0001} &0.0290 &0.9955 &\textbf{0.9272} & 0.0001 &0.0392 &0.9912 &0.9291 &0.0002 &0.0448 &0.9914 &0.9271\\
     MElastic & 0.0002 &0.0481 &0.9990 & 0.7376& 0.0002& 0.0553 & 0.9978& 0.7419 &0.0002 &0.0593 &\textbf{0.9970} &0.7442\\
     SRRR & \textbf{0.0001} &0.0151 &0.9965 &0.8354 &\textbf{0.0001} &0.0250 & 0.9905& 0.8229&\textbf{0.0001} &0.0297 & 0.9882 &0.7760\\
     ERRR & \textbf{0.0002} &0.0377 &\textbf{0.9995} &0.3763& 0.0002 &0.0440 & \textbf{0.9988}& 0.3742& 0.0002 &0.0539& 0.9958& 0.9026\\
     \hline
\end{tabular}}
\end{table}
\begin{table}[H]
\centering
\caption{Results of the numerical simulations for the case with a group structure when the number of explanatory variables is $p = 400$}
\resizebox{1.0\textwidth}{!}{
\begin{tabular}{|l|cccc|cccc|cccc|}
\hline
 & \multicolumn{4}{c|}{$n,n_{\rm test}=100,\ \tau=0.1$}
 & \multicolumn{4}{c|}{$n,n_{\rm test}=100,\ \tau=0.2$}
 & \multicolumn{4}{c|}{$n,n_{\rm test}=100,\ \tau=0.25$} \\ \hline
 \textbf{\textit{Method}} & $\bm{\mathit{MSE}_{\bm{B}}}$ & $\bm{\mathit{MSE}_{\bm{Y}}}$ & \textbf{\textit{TPR}} & \textbf{\textit{TNR}} & $\bm{\mathit{MSE}_{\bm{B}}}$ & $\bm{\mathit{MSE}_{\bm{Y}}}$ & \textbf{\textit{TPR}} & \textbf{\textit{TNR}} & $\bm{\mathit{MSE}_{\bm{B}}}$ & $\bm{\mathit{MSE}_{\bm{Y}}}$ & \textbf{\textit{TPR}} & \textbf{\textit{TNR}} \\\hline
  Proposed & 0.0006& \textbf{0.1757} &0.9332 &0.8146& 0.0009& \textbf{0.2347} &0.7965 &0.9478 &0.0012 &\textbf{0.2568} &0.7569 &0.9471\\
     MLasso & \textbf{0.0004} &0.2279 &0.9092 &\textbf{0.9433}& 0.0007& 0.2915& 0.8222 &\textbf{0.9576} &0.0009 &0.3071 &0.8124 &0.9152\\
     MElastic & \textbf{0.0004}& 0.2691 & \textbf{0.9848}& 0.8976& 0.0006& 0.2789& \textbf{0.9682} &0.8705 &0.0006 &0.2811& \textbf{0.9666}& 0.8791\\
     SRRR & 0.0005 &0.2002 &0.8702 &0.8579 &0.0007 &0.2256 &0.7543 &0.9274
     &0.0009 &0.3212 &0.7622 &0.8364\\
     ERRR & \textbf{0.0004} &0.2261 & 0.9828 &0.9172& \textbf{0.0005} &0.2427 &0.9615 &0.9397 & \textbf{0.0005} &0.2624 & 0.9582 &\textbf{0.9718}\\
\hline
  & \multicolumn{4}{c|}{$n,n_{\rm test}=500,\ \tau=0.1$}
 & \multicolumn{4}{c|}{$n,n_{\rm test}=500,\ \tau=0.2$}
 & \multicolumn{4}{c|}{$n,n_{\rm test}=500,\ \tau=0.25$} \\ \hline  
     \textbf{\textit{Method}} & $\bm{\mathit{MSE}_{\bm{B}}}$ & $\bm{\mathit{MSE}_{\bm{Y}}}$ & \textbf{\textit{TPR}} & \textbf{\textit{TNR}} & $\bm{\mathit{MSE}_{\bm{B}}}$ & $\bm{\mathit{MSE}_{\bm{Y}}}$ & \textbf{\textit{TPR}} & \textbf{\textit{TNR}} & $\bm{\mathit{MSE}_{\bm{B}}}$ & $\bm{\mathit{MSE}_{\bm{Y}}}$ & \textbf{\textit{TPR}} & \textbf{\textit{TNR}} 
     \\ \hline
   Proposed & \textbf{0.0001} &\textbf{0.0518} &0.9860 &0.7938& \textbf{0.0001}& \textbf{0.0588}& \textbf{0.9959}& 0.8987 & \textbf{0.0001}& \textbf{0.0802} &\textbf{0.9960}& \textbf{0.9125}\\
     MLasso & 0.0001 &0.0717 &0.9805 &\textbf{0.9582} &0.0002& 0.0992 &0.9598& \textbf{0.9616}
    & 0.0003 &0.1117 &0.9627 &0.7035\\
     MElastic &  0.0002& 0.0932 &0.9935& 0.8229& 0.0002 &0.1060& 0.9880 &0.8354 & 0.0003 &0.1135 &0.9856 &0.8392\\
     SRRR & \textbf{0.0001} &0.0577 & 0.9805 &0.7431 &0.0002 &0.0782& 0.9486 & 0.8047
     &0.0002 &0.0872 &0.9843& 0.8392\\
     ERRR & 0.0002 &0.0774 & \textbf{0.9952} &0.7143&0.0002& 0.0916 &0.9884 & 0.7232
     &0.0003& 0.1005 & 0.9843 &0.7262\\
    \hline
       & \multicolumn{4}{c|}{$n,n_{\rm test}=1000,\ \tau=0.1$}
 & \multicolumn{4}{c|}{$n,n_{\rm test}=1000,\ \tau=0.2$}
 & \multicolumn{4}{c|}{$n,n_{\rm test}=1000,\ \tau=0.25$} \\ \hline 
     \textbf{\textit{Method}} & $\bm{\mathit{MSE}_{\bm{B}}}$ & $\bm{\mathit{MSE}_{\bm{Y}}}$ & \textbf{\textit{TPR}} & \textbf{\textit{TNR}} & $\bm{\mathit{MSE}_{\bm{B}}}$ & $\bm{\mathit{MSE}_{\bm{Y}}}$ & \textbf{\textit{TPR}} & \textbf{\textit{TNR}} & $\bm{\mathit{MSE}_{\bm{B}}}$ & $\bm{\mathit{MSE}_{\bm{Y}}}$ & \textbf{\textit{TPR}} & \textbf{\textit{TNR}} 
     \\ \hline
    Proposed & \textbf{0.0001} &0.0315 &0.9930 & 0.8866& \textbf{0.0001} &0.0565 &\textbf{0.9986} &0.8668 &\textbf{0.0001} &0.0706 & \textbf{0.9990} &0.6647\\
     MLasso & \textbf{0.0001} &0.0399 &0.9910 & 0.9290 &  \textbf{0.0001} &0.0597& 0.9814 &\textbf{0.9317}&0.0002 &0.0686 &0.9763 &  \textbf{0.9322}\\
     MElastic & \textbf{0.0001}& 0.0759 & 0.9950 & \textbf{0.9835}& 0.0002 &0.0720 &0.9948 &0.7634& 0.0002 &0.0772 &0.9925& 0.7651\\
     SRRR &  \textbf{0.0001} &\textbf{0.0261}& 0.9882 &0.8882& \textbf{0.0001}& \textbf{0.0444}& 0.9761& 0.7992&\textbf{0.0001}& \textbf{0.0511} &0.9675 &0.8134\\
     ERRR & \textbf{0.0001} &0.0511& \textbf{0.9960} &0.9046 &  \textbf{0.0001}&0.0638& 0.9926 &0.9080&0.0002 &0.0696& 0.9897 &0.9097\\
     \hline
\end{tabular}}
\end{table}

\section{Real data application}
\quad In this Section, we applied the proposed method to a genetic dataset named "chin07" from package lol in R software \citep{chin2007high} in order to verify its usefulness.
The chin07 dataset consists of gene expression data from breast cancer tumor samples.
Specifically, it comprises mRNA expression data for important breast cancer genes and DNA copy number data for genomic regions containing candidate genes related to breast cancer.
The mRNA expression data consist of expression levels for seven gene regions mapped to key breast cancer genes across 106 samples.
On the other hand, the DNA copy number data consist of 339 genomic regions, constructed by aggregating adjacent DNA sequences using CGHregions, for the same 106 samples.
In addition, the DNA copy number data include 22 groups that classify these genomic regions.
\par{}
\quad The objective of the analysis using this dataset is to identify which genomic regions exhibit DNA copy number variations that influence mRNA expression levels and to what extent.
We perform regression analysis using the mRNA expression data as response variables and the DNA copy number data as explanatory variables.
\par{}
\quad Next, we present the validation procedure for the real data application.
The evaluation consists of the following 4 Steps.
In Step 1, the real dataset is divided into training data and test data.
Specifically, 80 samples are randomly selected from the dataset as the training data, and the remaining 26 samples are used as the test data.
In Step 2, for the training data, hyperparameters of the proposed method and comparative  methods that account for group structure are selected using 5-fold cross-validation.
In Step 3, each method with the selected hyperparameters is applied to the training data to estimate the matrix of regression coefficients.
In Step 4, predicted values of the response variables are obtained using the test data, and the mean squared prediction error (MSPE) between the predicted and true values is calculated.
Step 1 through Step 4 are repeated 100 times, and the mean of MSPE (Mean) and the standard error (SE) are computed to evaluate the performance of the proposed method.
\par{}
\quad Finally, we present the results of the real data application for the proposed method and the comparative methods.
The results are summarized in Table 5.
The proposed method achieved the smallest Mean, indicating higher predictive accuracy compared to the comparative methods.
This suggests that the proposed method more accurately captures the effects of DNA copy number variations on mRNA expression levels.
In addition, the proposed method exhibited smaller SE values, demonstrating relatively stable predictive accuracy.
These results indicate that the proposed method outperforms the comparative  methods in terms of both predictive accuracy and estimation stability.
In other words, by effectively capturing the association structure between DNA copy number variations and mRNA expression levels, the proposed method is shown to perform well in real data applications.
\begin{table*}[h]
\centering
\caption{Results of the real data application}
\begin{tabular*}{250pt}{@{\extracolsep\fill}|l|cc|@{\extracolsep\fill}}%
 \hline
 \textbf{\textit{Method}} & $\bm{\mathit{Mean}}$ & $\bm{\mathit{SE}}$ \\
 \hline
Proposed & \textbf{0.794} & \textbf{0.114}\\
MLasso & 0.838 & 0.119\\
MElastic & 0.865 & 0.122\\
SRRR & 0.813 & 0.125\\
ERRR & 0.847 & \textbf{0.114}\\
 \hline
\end{tabular*}
\end{table*}

\section{Conclusions}\label{sec5}
\quad In this paper, we proposed a novel method that incorporates pcLasso into the objective function of reduced-rank regression. 
The proposed method estimates regression coefficients by accounting for correlations among response variables while strongly biasing them toward principal component directions with large variance, and we confirmed that this leads to improved predictive accuracy. 
Furthermore, when explanatory variables exhibit a group structure, the proposed method was shown to further improve predictive accuracy by incorporating within-group correlation information.
\par{}
\quad The results of numerical simulations demonstrated that the proposed method generally outperformed the comparative  methods, and in particular, it successfully maintained predictive accuracy in scenarios with small sample sizes. 
For the comparative  methods, as the sample size decreases, the regression coefficients cannot be accurately estimated, resulting in a failure to maintain predictive accuracy.
This can be attributed to the fact that the proposed method introduces a regularization term that accounts for the correlation structure among explanatory variables and biases regression coefficients toward high-variance principal component directions. 
As a result, regression coefficients corresponding to information-rich principal components are emphasized, effectively suppressing estimation variance that tends to become unstable under small sample sizes, thereby maintaining predictive accuracy. 
Moreover, even in scenarios where explanatory variables have a group structure and each group exhibits a low-rank structure, the proposed method appropriately captured within-group correlation structures and biased regression coefficients along the corresponding principal component directions, achieving higher predictive accuracy than the comparative  methods.
However, for the accuracy of variable selection measured by the true negative rate (TNR), there were scenarios in which the comparative  methods outperformed the proposed method, indicating a remaining challenge. 
This is because the proposed method applies regularization that biases regression coefficients toward principal component directions with large variance while accounting for correlations among explanatory variables. 
Consequently, variables that do not truly contribute to predictive accuracy may still be selected if they are highly correlated with relevant variables. 
Indeed, even in numerical simulations on pcLasso conducted within the framework of a univariate linear regression model, the selection accuracy in terms of TNR is considerably lower than that of the true positive TPR \citep{tay2021principal}.
Compared with comparative  methods that more strongly enforce sparsity, the proposed method may therefore be less effective at removing irrelevant variables, leading to inferior TNR in some scenarios.
\par{}
\quad In the real data application, the proposed method was shown to more accurately capture the association structure between DNA copy number variations and mRNA expression levels compared with the comparative  methods. 
In addition, the proposed method achieved smaller values of ${\rm SE}$, indicating more stable predictive accuracy. 
These results suggest that it can be effectively applied to real-world data.
\par{}
\quad There are three issues to be addressed in the future research. 
First, in the numerical simulations, the rank $r$ of the matrix of regression coefficients, which is a tuning parameter, is fixed at three. It is necessary to examine how the predictive accuracy varies with different choices of the rank. In the simulation studies, each method was estimated assuming that the true rank of the matrix of regression coefficients was known; however, in practice, an appropriate rank should be selected in a data-driven manner, for example, via cross-validation. In particular, it is an important future task to more flexibly evaluate the predictive accuracy of the proposed method when the rank increases or when the true rank is misspecified. 
Regarding the choice of the rank, within the framework of SRRR, an appropriate rank is selected using a scree plot of cross-validation errors \citep{chen2012sparse}. 
Such criteria for selecting the number of latent variables can also be investigated within the proposed method.
Second, the proposed method currently assumes a non-overlapping group structure, in which each explanatory variable belongs to exactly one group. However, in real data analyses, such as gene expression studies, a single gene may belong to multiple biological pathways or functional groups, resulting in overlapping group structures. pcLasso addresses this situation by incorporating an overlapping group lasso penalty \citep{jacob2009group}. 
Extending the proposed method to allow overlapping group structures would enable analyses that better reflect realistic data settings and represents an important direction for future research.
Third, there remains an issue regarding the variable selection performance of the proposed method. 
Although the numerical simulation results demonstrate that the proposed method achieves higher predictive accuracy than the comparative methods, there are scenarios in which its selection accuracy can be further improved.

\bibliographystyle{apalike}  
\bibliography{references}  

@article{liu2022robust,
  title={Robust sparse reduced-rank regression with response dependency},
  author={Liu, Wenchen and Liu, Guanfu and Tang, Yincai},
  journal={Symmetry},
  volume={14},
  number={8},
  pages={1617},
  year={2022},
  publisher={MDPI}
}

@article{friedman2007pathwise,
  title={Pathwise coordinate optimization},
  author={Friedman, Jerome and Hastie, Trevor and H{\"o}fling, Holger and Tibshirani, Robert},
  volume={1},
  pages={302--332},
  journal={The Annals of Applied Statistics},
  year={2007}
}

@article{hilafu2020sparse,
  title={Sparse reduced-rank regression for integrating omics data},
  author={Hilafu, Haileab and Safo, Sandra E and Haine, Lillian},
  journal={BMC Bioinformatics},
  volume={21},
  number={1},
  pages={283},
  year={2020},
  publisher={Springer}
}

@article{chen2012sparse,
  title={Sparse reduced-rank regression for simultaneous dimension reduction and variable selection},
  author={Chen, Lisha and Huang, Jianhua Z},
  journal={Journal of the American Statistical Association},
  volume={107},
  number={500},
  pages={1533--1545},
  year={2012},
  publisher={Taylor \& Francis}
}

@article{tay2021principal,
  title={Principal component-guided sparse regression},
  author={Tay, Jingyi K and Friedman, Jerome and Tibshirani, Robert},
  journal={Canadian Journal of Statistics},
  volume={49},
  number={4},
  pages={1222--1257},
  year={2021},
  publisher={Wiley Online Library}
}

@article{chin2007high,
  title={High-resolution aCGH and expression profiling identifies a novel genomic subtype of ER negative breast cancer},
  author={Chin, Suet F and Teschendorff, Andrew E and Marioni, John C and Wang, Yanzhong and Barbosa-Morais, Nuno L and Thorne, Natalie P and Costa, Jose L and Pinder, Sarah E and Van de Wiel, Mark A and Green, Andrew R and others},
  journal={Genome Biology},
  volume={8},
  number={10},
  pages={R215},
  year={2007},
  publisher={Springer}
}

@inproceedings{ogutu2014regularized,
  title={Regularized group regression methods for genomic prediction: Bridge, MCP, SCAD, group bridge, group lasso, sparse group lasso, group MCP and group SCAD},
  author={Ogutu, Joseph O and Piepho, Hans-Peter},
  booktitle={BMC Proceedings},
  volume={8},
  number={Suppl 5},
  pages={S7},
  year={2014},
  organization={Springer}
}

@book{adachi2016matrix,
  title={Matrix-Based Introduction to Multivariate Data Analysis},
  author={Adachi, Kohei},
  year={2016},
  publisher={Springer}
}

@article{anderson1951estimating,
  title={Estimating linear restrictions on regression coefficients for multivariate normal distributions},
  author={Anderson, Theodore Wilbur},
  journal={The Annals of Mathematical Statistics},
  pages={327--351},
  year={1951},
  publisher={JSTOR}
}

@article{izenman1975reduced,
  title={Reduced-rank regression for the multivariate linear model},
  author={Izenman, Alan Julian},
  journal={Journal of Multivariate Analysis},
  volume={5},
  number={2},
  pages={248--264},
  year={1975},
  publisher={Elsevier}
}

@book{izenman2008modern,
  title={Modern Multivariate Statistical Techniques},
  author={Izenman, Alan J},
  volume={1},
  year={2008},
  publisher={Springer}
}

@book{reinsel1998multivariate,
  title={Multivariate Reduced-Rank Regression},
  author={Reinsel, Gregory C and Velu, Raja P},
  year={1998},
  publisher={Springer}
}

@article{yuan2006model,
  title={Model selection and estimation in regression with grouped variables},
  author={Yuan, Ming and Lin, Yi},
  journal={Journal of the Royal Statistical Society Series B: Statistical Methodology},
  volume={68},
  number={1},
  pages={49--67},
  year={2006},
  publisher={Oxford University Press}
}

@article{kobak2021sparse,
  title={Sparse reduced-rank regression for exploratory visualisation of paired multivariate data},
  author={Kobak, Dmitry and Bernaerts, Yves and Weis, Marissa A and Scala, Federico and Tolias, Andreas S and Berens, Philipp},
  journal={Journal of the Royal Statistical Society Series C: Applied Statistics},
  volume={70},
  number={4},
  pages={980--1000},
  year={2021},
  publisher={Oxford University Press}
}

@article{mukherjee2011reduced,
  title={Reduced rank ridge regression and its kernel extensions},
  author={Mukherjee, Ashin and Zhu, Ji},
  journal={Statistical Analysis and Data Mining: the ASA Data Science Journal},
  volume={4},
  number={6},
  pages={612--622},
  year={2011},
  publisher={Wiley Online Library}
}

@article{zou2005regularization,
  title={Regularization and variable selection via the elastic net},
  author={Zou, Hui and Hastie, Trevor},
  journal={Journal of the Royal Statistical Society Series B: Statistical Methodology},
  volume={67},
  number={2},
  pages={301--320},
  year={2005},
  publisher={Oxford University Press}
}

@article{chen2010graph,
  title={Graph-structured multi-task regression and an efficient optimization method for general fused lasso},
  author={Chen, Xi and Kim, Seyoung and Lin, Qihang and Carbonell, Jaime G and Xing, Eric P},
  journal={arXiv preprint arXiv:1005.3579},
  year={2010}
}

@book{gower2004procrustes,
  title={Procrustes Problems},
  author={Gower, John C and Dijksterhuis, Garmt B},
  volume={30},
  year={2004},
  publisher={Oxford university press}
}

@article{carroll1970analysis,
  title={Analysis of individual differences in multidimensional scaling via an N-way generalization of “Eckart-Young” decomposition},
  author={Carroll, J Douglas and Chang, Jih-Jie},
  journal={Psychometrika},
  volume={35},
  number={3},
  pages={283--319},
  year={1970},
  publisher={Springer-Verlag}
}

@article{harshman1970foundations,
  title={Foundations of the PARAFAC procedure: Models and conditions for an “explanatory” multi-modal factor analysis},
  author={Harshman, Richard A and others},
  journal={UCLA working papers in phonetics},
  volume={16},
  number={1},
  pages={84},
  year={1970},
  publisher={Los Angeles, CA}
}

@inproceedings{jacob2009group,
  title={Group lasso with overlap and graph lasso},
  author={Jacob, Laurent and Obozinski, Guillaume and Vert, Jean-Philippe},
  booktitle={Proceedings of the 26th Annual International Conference on Machine Learning},
  pages={433--440},
  year={2009}
}

\end{document}